\documentclass[11pt,  oneside, a4paper]{article}
\usepackage[utf8]{inputenc}
\usepackage[a4paper, total={6.5in,9in}]{geometry}

\usepackage{graphicx} 
\usepackage{bbm}
\usepackage{amsfonts}
\usepackage{enumitem}
\usepackage{fullpage}
\usepackage{amsmath}
\usepackage{hyperref}
\usepackage{mathtools}
\usepackage{lipsum}
\usepackage{blindtext}
\usepackage{titlesec}
\usepackage{amssymb}
\usepackage{amsthm}
\usepackage{rotating}
\usepackage{subcaption}
\usepackage{caption}
\usepackage{svg}
\usepackage{xcolor}
\usepackage{verbatim}
\usepackage{natbib}

\newtheorem{theorem}{Theorem}[section]

\newtheorem{lemma}[theorem]{Lemma}
\newtheorem{claim}[theorem]{Claim}
\newtheorem{corollary}[theorem]{Corollary}
\theoremstyle{definition}

\newtheorem{aspt}[theorem]{Assumption}

\newcommand{\footremember}[2]{%
    \footnote{#2}
    \newcounter{#1}
    \setcounter{#1}{\value{footnote}}%
}

\newcommand{\nrev}{n_{rev}}
\newcommand{\crev}{c_{rev}}
\newcommand{\tldt}{\tilde{t}}
\newcommand{\tldtau}{\tilde{\tau}}
\newcommand{\E}{\mathbb{E}}
\newcommand{\cP}{\mathbb{P}}

\newcommand{\nold}{n_1}

\title{When Should a Principal Delegate to an Agent in Selection Processes?}
\author{Benjamin Fish\footremember{alley}{University of Michigan. Email: benfish@umich.edu} 
\and Diptangshu Sen\footremember{trailer}{Georgia Institute of Technology. Email: dsen30@gatech.edu}
\and Juba Ziani\footremember{somethingelse}{Georgia Institute of Technology. Email: jziani3@gatech.edu}
}
\date{\today}

\begin{document}

\maketitle

\begin{abstract}
Decision-makers in high-stakes selection processes often face a fundamental choice: whether to make decisions themselves or to delegate authority to another entity whose incentives may only be partially aligned with their own. Such delegation arises naturally in settings like graduate admissions, hiring, or promotion, where a principal (e.g. a professor or worker) either reviews applicants personally or decisions are delegated to an agent (e.g. a committee or boss) that evaluates applicants efficiently, but according to a potentially misaligned objective. 

We study this trade-off in a stylized selection model with noisy signals.  The principal incurs a cost for selecting applicants, but can evaluate applicants based on their fit with a project, team, workplace, etc. In contrast, the agent evaluates applicants solely on the basis of a signal that correlates with the principal's metric, but this comes at no cost to the principal. 
Our goal is to characterize when delegation is beneficial versus when decision-making should remain with the principal. We compare these regimes along three dimensions: (i) the principal’s utility; (ii) the quality of the selected applicants according to the principal's metric; and (iii) the fairness of selection outcomes under disparate signal qualities.
\end{abstract}

\section{Introduction}\label{sec:intro}

With the advent of automated decision making, software as a service, and agentic AI, there are many potential benefits to delegating to another decision maker.  Third parties can offer AI and other tools that promise cost-efficient and accurate decision making at scale, such as for auto-bidding in online auctions~\citep{aggarwal2024auto}, self-driving cars~\citep{autonomous}, algorithmic pricing ~\citep{pricing}, and resume-screening agents in hiring~\citep{aihiring}.  This includes tools for selection processes, where the goal is to choose from a pool of applicants, including in hiring, admissions, banking, or creating data sets.

Importantly, many of these decisions are made under time, cost, expertise, and informational constraints and inefficiencies, providing incentives for a decision maker to delegate to a third-party agent or AI tool. When the decision maker, whom we will call the \emph{principal}, delegates a selection process to another, the \emph{agent}, the agent can specialize in that process and take advantage of more data, larger scale, and higher efficiency, while the principal reaps the benefits of the selected applicants.  The agent may seek control over decision making anyway, for the sake of centralization \& standardization, enforcing constraints like fairness considerations, or ensuring that the decisions benefit themselves \citep{brynjolfsson2023centralize,kapor2024aftermarket,LUPIA201558}. However, delegation comes with its own set of challenges. The principal and the agent may not be aligned in terms of objectives or preferences \citep{alur2024human,CANDRIAN2022107308,LUPIA201558}, leading to lower quality outcomes for the principal \citep{sliwka2001delegation}, a loss of agency \citep{LUPIA201558}, and an opportunity for the agent to take advantage of their control over the process \citep{JENSEN1976305}. Delegation to AI hiring tools, for example, may ignore cultural fit, reduce interaction between candidates and hiring managers, and fail to properly evaluate candidate's attributes \citep{aizenberg2025examining,kim2025ai}.  Brokers who do not select suitable investments for their clients risk violating legal requirements \citep{anderson1992suitability}. Delegation may also lead to issues when it comes to responsible and fair decision-making: for example, Amazon's AI hiring tool was found to overwhelmingly favor male applicants to female applicants purely because it had been trained on historical data of successful applicant resumes which were predominantly male~\citep{amazon}, failing to correct for biases in its own data. 
 
This naturally motivates the following question: \textbf{\textit{what are the consequences of a principal delegating decisions to an agent on both quality and fairness of decisions?}} To address this question, we study delegation specifically in the context of \emph{selection processes}. Selection problems arise frequently in hiring and admissions and play a key role in organizational decision-making. We consider a stylized principal–agent model of selecting applicants whose quality can only be evaluated using noisy signals.  Applicants are characterized by two dimensions: (i) an ability score, reflecting their underlying skill or level of preparation, and (ii) a fit score, capturing how well they align with a specific team, project, or advisor. The principal (e.g., a team in hiring settings or a professor in graduate admissions) incurs a cost when selecting applicants but can evaluate them along both dimensions. The agent (e.g., a central admissions committee or an automated decision tool), by contrast, evaluates applicants solely on ability, but at no direct cost to the principal. 

\paragraph{Summary of Contributions.} Our main contributions are as follows:
\vspace{-1mm}
\begin{enumerate}
    \item In Section~\ref{sec:model}, we introduce our stylized principal-agent model where a principal (they/them) must decide which applicants to select from a pool of applicants, but can delegate the decision to an agent (it).
    \item We provide theoretical characterizations of the utility achieved \emph{by the principal}, both when they make decisions themselves (Section~\ref{sec:decentralization}) and when they delegate to the agent (Section~\ref{sec:centralization}). 
    \item We expand our analysis to selecting applicants from \emph{multiple} populations. In particular, in our model, the populations are \emph{intrinsically} identical; they differ \emph{only} in the accuracy of the noisy signal observed about each population. We consider two models of such disparity: i) one where applicants in the disadvantaged groups have a \emph{negatively biased} signal of their quality and ii) one where they have a \emph{noisier} signal, compared to the advantaged group. We study the fairness implications of signal disparity in the context of demographic parity, under the the delegated model in Section~\ref{sec:centralization} and the non-delegated model in Section~\ref{sec:decentralization}.
    \item Finally, in Section~\ref{sec:comparison}, we synthesize our results and provide insights on when delegation is beneficial, comparing the outcomes of both decisions across different metrics: i) quality of selected applicants, ii) the principal's utility (including both quality \emph{and} selection costs), and iii) fairness defined as demographic parity. We make several interesting (and counter-intuitive) observations: there may be non-monotonicities in the principal's preferences for delegating vs not delegating in terms of how much they weight ``ability'' vs ``fit''. 
    Further, what is preferable for fairness heavily depends on whether signal disparities come from i) systematic under-estimation of one population vs ii) disparate informativeness (e.g., variance) of signals across populations. 
\end{enumerate}

\paragraph{Related Work.} Our work is closely related to an extensive body of research in economics on principal-agent problems~\citep{ross1973economic}. In traditional principal-agent problems, a principal and an agent usually have asymmetric information access and misaligned interests; the goal for the principal is then to design a mechanism so as to incentivize the strategic agent to act in the principal's interest. While our work has similarities with this framework (misaligned objectives between principal and agent and information asymmetry---our agent cannot observe ``fit''), there are also major differences: in our setting, the principal cannot influence the decisions of the agent,  
and the principal only gets to decide whether it is beneficial for them to delegate or not. For a detailed review on different aspects of generalized principal-agent problems like contract design, information design, learning etc., we refer readers to~\citep{dutting2024algorithmic,dughmi2017algorithmic,lin2024generalized}.  

We focus on selection problems, with motivating examples arising in hiring and admissions. Designing efficient and high quality hiring processes has been the focus of long-standing research. Prominent directions include analysis of popular heuristics like hiring `above the mean' or `above the median'~\citep{broder2010hiring,helmi2013analysis}, algorithmic hiring under uncertainty~\citep{purohit2019hiring,li2025hiring}, screening under sequential or pipeline settings~\citep{cohen2023sequential,epstein2024selection}, to name a few. Our work contributes to this broad area, and specifically studies the problem of delegation in the context of hiring and selection problems where the principal has the option to outsource the selection process entirely to the agent. 
 
With the recent proliferation of AI tools that are capable of sophisticated decision-making, this type of outsourcing or \emph{delegation} is becoming increasingly common. In turn, recent works have focused on whether decision tasks in high-stakes settings should be assigned to `algorithmic delegates' that can act reliably without human oversight~\citep{milewski1997delegating,lubars2019ask}, and more broadly on the design of synergistic systems with both human and AI collaborators~\citep{bansal2019beyond,bansal2021most,lai2022human,de2024towards,donahue2024two}. Recent work has also looked at the task of selecting the most optimal delegate among many~\citep{greenwood2025designing,raghavan2024competition}. Our focus is slightly different, mostly centering on delegation in \emph{selection processes}.

Finally, our work has implications for responsible AI and fairness. Since the origins of the theory of statistical discrimination \citep{arrow1971theory,phelps1972statistical}, researchers have contributed across multiple dimensions of fairness, including but not limited to the  challenges of algorithmic fairness \citep{kleinberg2018algorithmic,green2022escaping,corbett2017algorithmic,kleinberg2018inherent}, fairness auditing \citep{kallus2022assessing}, fairness in ML\citep{hardt2016equality,karimi2021enhancing,mary2019fairness}, fairness for better societal outcomes\citep{mouzannar2019fair,pitoura2022fairness,zehlike2020matching,bower2017fair} etc. Fairness has also been studied extensively in the context of ranking and selection problems~\citep{zehlike2017fa,cohen2019efficient,blum2022multi,celis2016fair,celis2024centralized,kleinberg2018selection,emelianov2020fair,garg2021standardized}. Most of these works either look at the problem of fair selection in the presence of different forms of bias or consider the fairness implications of various selection strategies in the context of single-stage or multi-stage/sequential selection. On the contrary, we look at a different problem: we explore whether organizational decisions like who is conducting the selection process (principal or agent) can also have major fairness implications.

\section{Model}\label{sec:model}

This section provides our model: the applicants, the principal, and the two different types of selection processes (one where the principal makes all decisions and the other where the principal delegates to the agent). We include an extension to a multi-group setting (with an ``advantaged'' and a``disadvantaged'' group) to understand the fairness implications of selection processes with and without delegation. A key feature of our model is that applicants can only be evaluated using noisy or imperfect signals. Note that \textit{we take the perspective of the principal, when we talk about utilities or quality of selected applicants}. However, \textit{our fairness discussion is from a group perspective\footnote{There can be many fairness perspectives like individual fairness, group fairness, procedural fairness etc.}~\cite{dwork2012fairness}.} 

\subsection{Applicant Characteristics \& Principal Incentives}

\paragraph{Applicants} An applicant is described by a set of attributes $(s, \tilde{s}, f, \tilde{f})$. The tuple ($s, f$) indicates the applicant's \emph{unobservable} private type. ($\tilde s, \tilde f$) represents noisy signals about $s$ and $f$ respectively---while $\tilde s$ is observable to everyone, $\tilde f$ can only be observed by the principal after careful evaluation (which may incur a cost). For ease of exposition, we can think of $s$ as the applicant's \emph{true ability}---i.e., how prepared he/she is for the job or program. This is unobservable but, when applying, the applicant provides a noisy estimate of their ability in the form of test scores or college transcripts ($\tilde s$). Similarly, we can think of $f$ as a \textit{fit score} $f$---how well-matched they are for the job or program. While $f$ again cannot be observed exactly, the principal has the expertise to determine a noisy estimate $\tilde f$ by reviewing the application.

\paragraph{Principal incentives} The principal is interested in both the applicant's type $s$ and their fit $f$. So, irrespective of whether they are delegating to the agent or making decisions on their own, the utility of any given selected applicant to the principal is given by:
\[
t \triangleq \alpha f + (1-\alpha) s,
\]
where $\alpha$ denotes how much the principal cares about fit versus type. We call $t$ the applicant's \emph{quality according to the principal's metric} or often for convenience, simply \textit{quality}.

\paragraph{Key Assumptions.} Throughout this work, we make the following assumptions about the distribution of applicants:

\begin{aspt}[i.i.d. students]\label{aspt:Gaussian}
Applicant types $(s,f)$ are drawn i.i.d. $f \sim \mathcal{N}(0,\sigma_f)$ and $s \sim \mathcal{N}(0,\sigma_s)$.
\end{aspt}

\begin{aspt}[Gaussian estimates for fit and scores]\label{aspt:gauss_noise}
We have that $\tilde{f} = f + \epsilon_f$ and $\tilde{s} = s + \epsilon_s$ where $\epsilon_f$ and $\epsilon_s$ are mean-zero Gaussian noise. Further, $f,~s,~\epsilon_f,~\epsilon_s$ are mutually independent. 
\end{aspt}
Thus, the noisy signals $\tilde f$ and $\tilde s$ are also mean-zero Gaussian random variables. We define $\sigma_{\tilde f}$ and $\sigma_{\tilde s}$ as the variances of $\tilde f$ and $\tilde s$ respectively. 
Gaussian assumptions are relatively standard when it comes to modeling population distributions. We refer to~\cite{kannan2019downstream,garg2020dropping,liu2021test} for a few examples of works with similar assumptions. Based on the above assumptions, $t \sim \mathcal{N}(0, \sigma_{t})$ with $\sigma_{t} = \sqrt{\alpha^2 \sigma_f^2 + (1-\alpha)^2 \sigma_s^2}$, and similarly $\tilde t \sim \mathcal{N}(0, \sigma_{\tilde t})$ where $\sigma_{\tldt}$ depends on $\sigma_{\tilde f}$ and $\sigma_{\tilde s}$. We can interpret $\sigma_t$ as the extent of diversity in applicant quality in the population and $\sigma_{\tldt}$ would be its noisy counterpart.

\subsection{Types of Selection Processes}
We consider two types of selection processes in our model: \textit{one where the principal delegates to the agent} and \textit{the other where the principal makes all decisions themselves}. For example, in the context of graduate admissions, the first type of selection involves all admission decisions being made by a central admissions committee without the direct involvement of the professor (therefore the professor does not incur any costs, but their utility depends on the average quality of students hired by the committee). In the latter type of selection, the professor themselves are tasked with the responsibility of reviewing applications (incurring a cost for every application reviewed) and deciding which students to hire. We now elaborate on each type of selection process in the general setting.

\paragraph{Selection Process Delegated to Agent.} In this setting, the agent makes decisions unilaterally on which applicants to select. The agent does not see the fit score $\tilde{f}$ and bases its evaluation entirely on the perceived preparedness of the applicant, $\tilde{s}$. The decision rule is straightforward: \textit{admit applicants with $\tilde{s} \geq \tau_1$}, where $\tau_1$ is a pre-determined threshold. In the context of graduate admissions, this could be, for example, a GPA or GRE score requirement for admissions.

The principal's overall \emph{ex-ante} expected utility when $k$ applicants are hired by the agent is then given by: 
\begin{align}\label{eq:cent_util} 
        U_{dg}(\tau_1) = k \cdot \mathbb{E}\left[t~|~\tilde s \geq \tau_1 \right]. 
\end{align} 
Note here that \emph{ex-ante}, the principal only knows that the hired applicants would be above the threshold and may have no additional information about them. Importantly, when the selection process is delegated, the principal themselves incur no time or effort cost, but the downside is that some of these hired applicants may be poorly matched to the job thereby reducing the principal's overall expected utility.

\paragraph{Selection Process under No Delegation.}
In this setting, the principal wants to hire at most $k$ applicants and is tasked with the responsibility of selecting applicants themselves. Reviewing an application reveals the applicant's noisy fit score $\tilde f$ to the principal which enables determination of noisy estimate $\tldt$ of the applicant's quality,  
\[
\tldt = \alpha \tilde{f} + (1-\alpha) \tilde{s}.
\]
However, reviewing each application incurs a fixed marginal cost of $\crev$. The principal has to strategically make two decisions: i) what endogenous selection threshold $\tldtau$ to use (an applicant gets hired only if their perceived quality exceeds this threshold), and ii) how many applications $\nrev$ to review.  
In this case, if the principal reviews $\nrev$ applications, they end up selecting $\nrev \cdot \mathbb{P}[\tldt \geq \tldtau]$ applicants in expectation and each hired applicant earns them an expected utility of $E [t~|~\tldt \geq \tldtau]$. Therefore, the principal's overall expected utility from the selection process when they make decisions themselves is given by:
\begin{align}\label{eq:decent_util}
      U_{ndg}(\tldtau, \nrev) = \nrev \cdot \mathbb{P}[\tldt \geq \tldtau] \cdot \mathbb{E}[t~|~\tldt \geq \tldtau] - \nrev \cdot \crev,
\end{align}
where $\nrev \cdot \mathbb{P}[\tldt \geq \tldtau] \leq k$ (the principal does not want to hire more than $k$ applicants in expectation). 

Before concluding this segment, we specify the following assumption that we make throughout the paper:
\begin{aspt}\label{aspt:large}
Our model always operates in the regime where the applicant pool is sufficiently large. This ensures that even for high thresholds $\tau_1$ and $\tldtau$, sufficiently many candidates can always be found above the threshold.
\end{aspt} 
We expect this assumption to hold true in many real-world settings like graduate admissions at top US universities or top industry companies which receive hundreds to thousands of applications every year for a handful of positions.

\subsection{Fairness under Multiple Groups} 

In order to explore the effects of the choice of selection process on fairness, beyond the general setting with a single homogeneous group of applicants, we also consider a multi-group setting. In particular, we assume that there are two groups $A$ and $B$\footnote{Our results qualitatively extend beyond two groups, but for the sake of notational ease, our exposition assumes there are exactly two groups.} with demographic ratios $\Lambda_A = \lambda$ and $\Lambda_B = 1-\lambda$ respectively for some $\lambda \in (0,1)$, i.e. a uniformly randomly chosen applicant from the applicant pool belongs to group $A$ with probability $\lambda$.

\begin{aspt}\label{aspt:equal}
Both groups have the same distribution of true types ($s, f$). 
\end{aspt}
This assumption is based on the \textit{we are all equal} (WAE) worldview introduced in the seminal work of \cite{friedler2021possibility}. This enables us to study group-level disparate outcomes in the selection process without a difference in true types as a possible cause of those disparate outcomes.

In our setting, applicants from Group $B$ are disadvantaged compared to those from Group $A$. The disadvantage is primarily with respect to how the noisy signal $\tilde s$ about applicant ability is perceived (the disadvantage can also manifest in $\tilde f$ but our main insights still go through unchanged) and can be of one of the two following forms: 
\begin{itemize}
    \item \textit{Signal $\tilde{s}$ with biased mean.}  
    The signal $\tilde s$ for group $A$ applicants is drawn from $\mathcal{N}(0, \sigma_{\tilde s} ^2)$ while for $B$ applicants $\tilde s$ is drawn from $\mathcal{N}(-\beta, \sigma_{\tilde s}^2)$ for some $\beta > 0$. Consequently, the signal distribution for group $A$ stochastically dominates the signal distribution for group $B$, i.e., given any threshold, the probability of finding an applicant from group $A$ above the threshold is higher than the corresponding probability for an applicant from group $B$.
    Such disparities are frequently observed in practice and are very well-documented, for example, disparities in SAT scores between high-income and low-income students~\citep{zwick2013sat} or gender gaps in math and verbal subject scores on standardized tests~\citep{griselda2024gender}. 
    \item \textit{Signal $\tilde{s}$ with disparate variance.} In this case, applicants from Group $B$ are less well-understood than those from group $A$. We model this as receiving higher variance signals in population $B$, in line with previous works like ~\citet{garg2020dropping,kannan2019downstream}. 
    The lack of understanding of population B is modeled as having $\sigma_{\tilde s, B} >  \sigma_{\tilde s, A}$. This may be, for example, because they have been historically marginalized in academia or the workplace, or because applicants are applying with backgrounds that decision-makers have little prior experience with. %
\end{itemize}
We aim to investigate how the type of selection process affects the group-wise composition of selected applicants---in particular, if the choice of whether to delegate or not has outsized impacts on disparity between the two groups. 

\section{A Selection Process Delegated to the Agent}\label{sec:centralization}

In this section, we consider the setting where the principal delegates the task of selecting applicants to the agent.
First, we explore the effect that delegation has on principal utilities in the single group setting: we focus on understanding how said utility evolves in the parameters of the problem, including the selection threshold $\tau_1$ and the importance of fit controlled by $\alpha$. We then consider the multi-group setting, in particular observing how the composition of admitted applicants in a selection process under delegation looks like and how it changes with the selection threshold and the extent of advantage one group already has over the other. 

\subsection{Single Group Setting: A Utilitarian View}
Recall that the principal's expected \textit{ex-ante} utility per hired applicant is given by $\mathbb{E}[t~|~\tilde s \geq \tau_1]$. Our first main result expresses this utility in closed form in terms of key problem parameters. 
\begin{lemma}\label{lem:tech_prof_util_cent}
When the selection process is delegated to the agent, the expected ex-ante utility per hired applicant earned by a principal who puts weight $\alpha \in (0, 1)$ on applicant fit, is given by:
\[
      \mathbb{E}[t~|~\tilde s \geq \tau_1] = \frac{(1-\alpha)\sigma_s^2}{\sigma_{\tilde s}}\cdot H\left(\frac{\tau_1}{\sigma_{\tilde s}} \right),
\]
where $\tau_1$ is the agent's selection threshold and $H(\cdot)$ is the hazard rate function of a standard normal random variable. 
\end{lemma}
The proof of the lemma can be found in Appendix~\ref{pf:lemma1}. A few direct consequences of this lemma are as follows:

\begin{corollary}\label{cor:tau1_monotone}
When the selection process is delegated to the agent, the principal's ex-ante utility from each admitted applicant is monotonically increasing in the agent's selection threshold $\tau_1$. 
\end{corollary}

This follows directly from Lemma~\ref{lem:tech_prof_util_cent}: the hazard rate function of the standard normal random variable is known to be monotonically increasing in its argument. We provide a short proof in Appendix~\ref{app:monotone_hazard}. The corollary shows that when the agent uses a higher selection threshold, the principal's utility increases: indeed, i) the average ability of accepted applicants increases and ii) the principal does not incur any selection cost. Second, we note that the principal's utility is also monotonic and decreasing in $\alpha$:

\begin{corollary}\label{cor:alpha_monotone}
When the selection process is delegated to the agent, the principal's ex-ante utility diminishes monotonically in $\alpha$ which is the principal's preferential weight on applicant fit.
\end{corollary}

This follows from Lemma~\ref{lem:tech_prof_util_cent}---increasing $\alpha$ decreases the principal's expected utility from each hired applicant since the principal's preferences increasingly diverge from the preferences of the agent they delegate to.

\subsection{Fairness Implications in Multi-Group Settings}

When the applicant pool consists of applicants from different groups, an immediate follow-up question is: what is the average group-wise composition of the set of applicants hired through a selection process which has been delegated to the agent? The answer to this question has important fairness implications. For example, if groups $A$ and $B$ are identical in all respects and group $A$ has a demographic ratio of $\lambda$, one \textit{fair outcome} might be that on average, $\lambda$ fraction of the admitted applicants come from group $A$. This subscribes to the well-known concept of \textit{demographic parity}\footnote{Note here that because our populations have identical distributions of true ability and fit, demographic parity is a natural notion of fairness.} in the fairness literature \citep{dwork2012fairness}. In our setting, however, despite being intrinsically identical, the observed signals about the two groups are not identical---in particular, group $B$ is disadvantaged in some way (either their signals are noisier (higher variance) or the distribution mean is negatively biased compared to Group $A$). Our goal in this section is to investigate how the extent of this disparity between the groups manifests in the outcome of the selection process. 

Importantly, \textit{the agent uses the same selection standard ($\tau_1$) for everyone, irrespective of group identity}---our agent is (naively) group-blind\footnote{For example, an automated hiring tool that relies on ostensibly neutral keywords may apply a uniform decision rule across applicants, while failing to account for the fact that such keywords can be disparately associated with different demographic groups; see for example Amazon's recent failure in using automated AI decision-making tools for hiring~\citep{amazon}. Or a central university admission committee that makes admissions decisions without taking into account the disparate meaning of standardized scores across different populations.}. A key difference of this with the setting where the principal does not delegate is that in the latter, the principal has the flexibility to customize group-specific selection standards---for example, in the context of graduate admissions, individual professors can still make their own decisions about who to hire.

\paragraph{Composition of selected applicants} 
Let $\Phi_A(\cdot)$ and $\Phi_B(\cdot)$ indicate the normal CDFs of the signal ($\tilde s$) distributions for groups $A$ and $B$ respectively. Now, let $\Phi_M(\cdot)$ indicate the CDF of the Gaussian mixture distribution where a sample belongs to group $A$ with probability $\lambda$ and group $B$ with probability $1-\lambda$: 
\[
     \Phi_M(x) = \lambda \cdot \Phi_A(x) + (1-\lambda)\cdot \Phi_B(x), \quad \forall~x \in \mathbb{R}. 
\]
Similarly, $\bar \Phi_A(\cdot)$, $\bar \Phi_B(\cdot)$ and $\bar \Phi_M(\cdot)$ indicate the corresponding complementary CDFs. We now present the following result (whose proof can be found in Appendix~\ref{pf:lem2}):

\begin{lemma}\label{lem:cent_comp}
Consider a selection process which is delegated to the agent with a selection threshold of $\tau_1$. For any applicant hired through this process from a mixed applicant pool, the probability that said applicant belongs to group $i$ is given by: 
\[
          \frac{\Lambda_i \cdot \bar \Phi_i(\tau_1)}{\bar \Phi_M(\tau_1)} \quad \forall~i \in \{A, B\},
\]
where $\Lambda_A = \lambda$ and $\Lambda_B = 1-\lambda$.
\end{lemma}

\paragraph{When do significant disparities arise?} 
In order to understand how \textit{fair} the realized composition of selected applicants is, we introduce the following fairness metric: 
\begin{align}\label{eq:fairness}
        \mathcal{D} = \frac{Y_A}{\Lambda_A} - \frac{Y_B}{\Lambda_B},
\end{align}
where $Y_i$ is the realized proportion of admits that belong to group $i$ (random variable) with $\mathbb{E}[Y_i] = \frac{\Lambda_i \cdot \bar \Phi_i(\tau_1)}{\bar \Phi_M(\tau_1)}$ (as shown in Lemma~\ref{lem:cent_comp}) and $\Lambda_i$ is the demographic ratio of group $i$ as before. The ratio $Y_i/\Lambda_i$ indicates whether group $i$ is over- or under-represented in the pool of hired applicants compared to what is \textit{demographically fair}. Thus, the metric $\mathcal{D}$ (which looks at the difference of these ratios) is a measure of the extent and the direction of unfairness in the composition of hires across the two groups. Note that when groups $A$ and $B$ have identical signal distributions, $\bar \Phi_A(\tau_1) = \bar \Phi_B(\tau_1) = \bar \Phi_M(\tau_1)$ which implies that $\mathbb{E}[\mathcal{D}] = 0$---this corresponds to the case where we have a \textit{perfectly fair} group composition in expectation. A larger absolute value of $\mathbb{E}[\mathcal{D}]$ indicates more unfairness, with positive and negative values indicating unfairness in favor and against the advantaged group ($A$) respectively. Our next result characterizes some of the statistical properties of the fairness metric $\mathcal{D}$: 

\begin{theorem}\label{thm:D_stats}
Consider a selection process which has been delegated to the agent and uses a selection threshold of $\tau_1$ to hire applicants from a mixed applicant pool consisting of groups $A$ and $B$. In that case, the group fairness metric $\mathcal{D}$ satisfies: 
\[
                  \mathbb{E}[\mathcal{D}] = \frac{\bar \Phi_A(\tau_1) - \bar \Phi_B(\tau_1)}{\bar \Phi_M(\tau_1)}, \quad \text{and} \quad
     \frac{\bar \Phi_A(\tau_1) + \bar \Phi_B(\tau_1)}{\bar \Phi_M(\tau_1)} \geq \mathbb{E}[\vert \mathcal{D} \vert ] \geq \frac{\vert \bar \Phi_A(\tau_1) - \bar \Phi_B(\tau_1)\vert}{\bar \Phi_M(\tau_1)}.
\]
\end{theorem}
The proof of the theorem can be found in Appendix~\ref{pf:thm1}. Firstly observe that when Group $B$ is disadvantaged, $\bar \Phi_A(\tau_1) \neq \bar \Phi_B(\tau_1)$ which implies that the lower bound on $\mathbb{E}[\vert\mathcal{D}\vert]$ is non-trivial. This implies that if the signal distributions are relatively different and one group has a larger tail than the other (for example, due to large additive bias), then significant disparities will arise between the groups on average.
The expression for $\mathbb{E}[\mathcal{D}]$ also gives us insights about the direction of the expected \textit{unfairness}, whether it is in favor of group $A$ or group $B$. In particular:

\begin{corollary}\label{corr:dir_unf}
Suppose, $\tilde s_A\sim \mathcal{N}\left(0, \sigma_{\tilde s, A}^2 \right)$ and $\tilde s_B\sim \mathcal{N}\left(-\beta, \sigma_{\tilde s, B}^2 \right)$ with $\frac{\sigma_{\tilde s, B}}{\sigma_{\tilde s, A}} = r$.  Then, 
\[
    \mathbb{E}[\mathcal{D}] >0 \quad \text{iff }  (r-1)\tau_1 -\beta < 0, \quad \text{and} \quad \mathbb{E}[\mathcal{D}] < 0 \quad \text{iff }  (r-1)\tau_1 - \beta > 0.
\]
\end{corollary}

\begin{figure}[!h]
    \centering
    \subfloat[$\Lambda_A = 0.50$]{\includegraphics[width=0.32\textwidth]{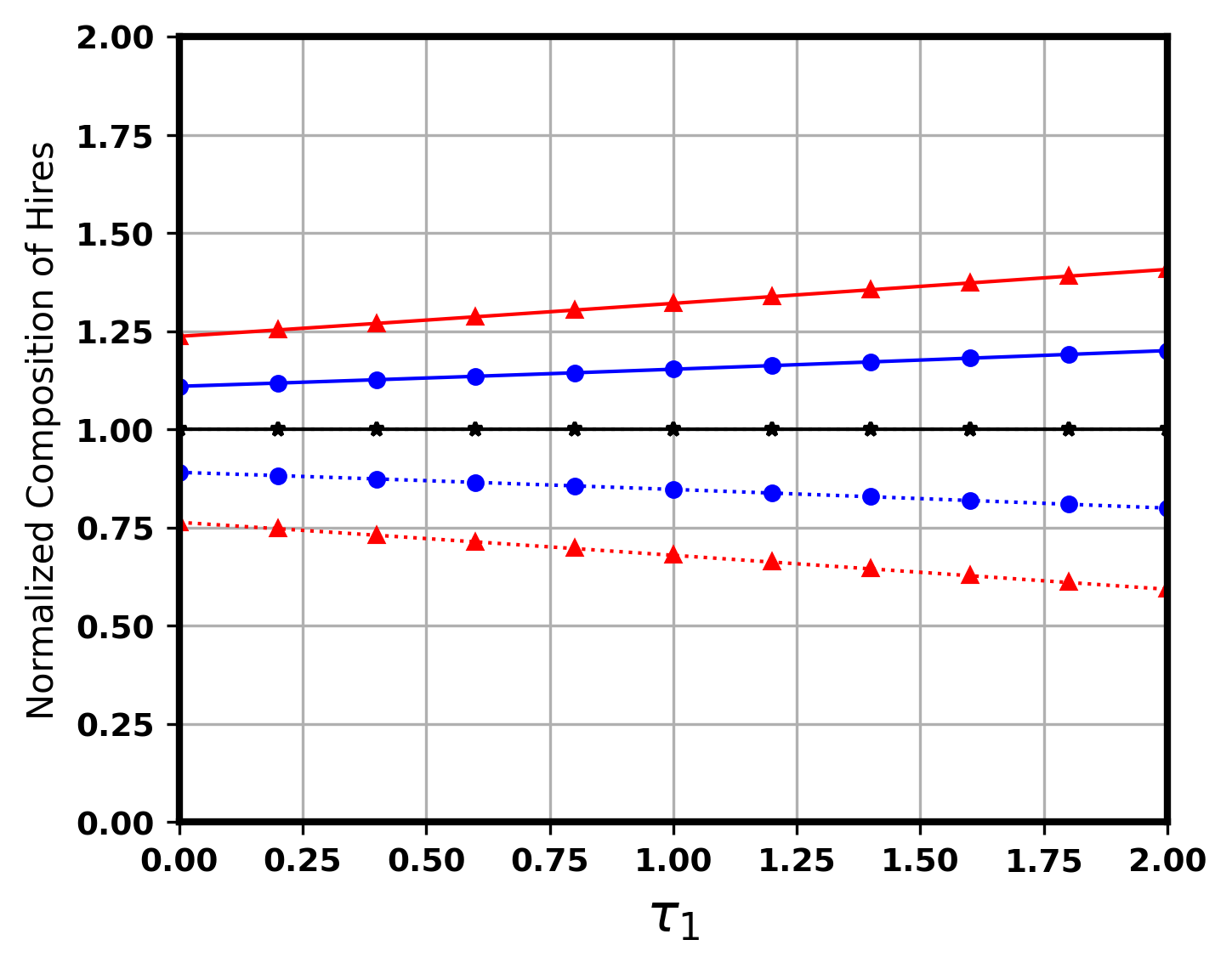}}
    \subfloat[$\Lambda_A = 0.65$]{\includegraphics[width=0.32\textwidth]{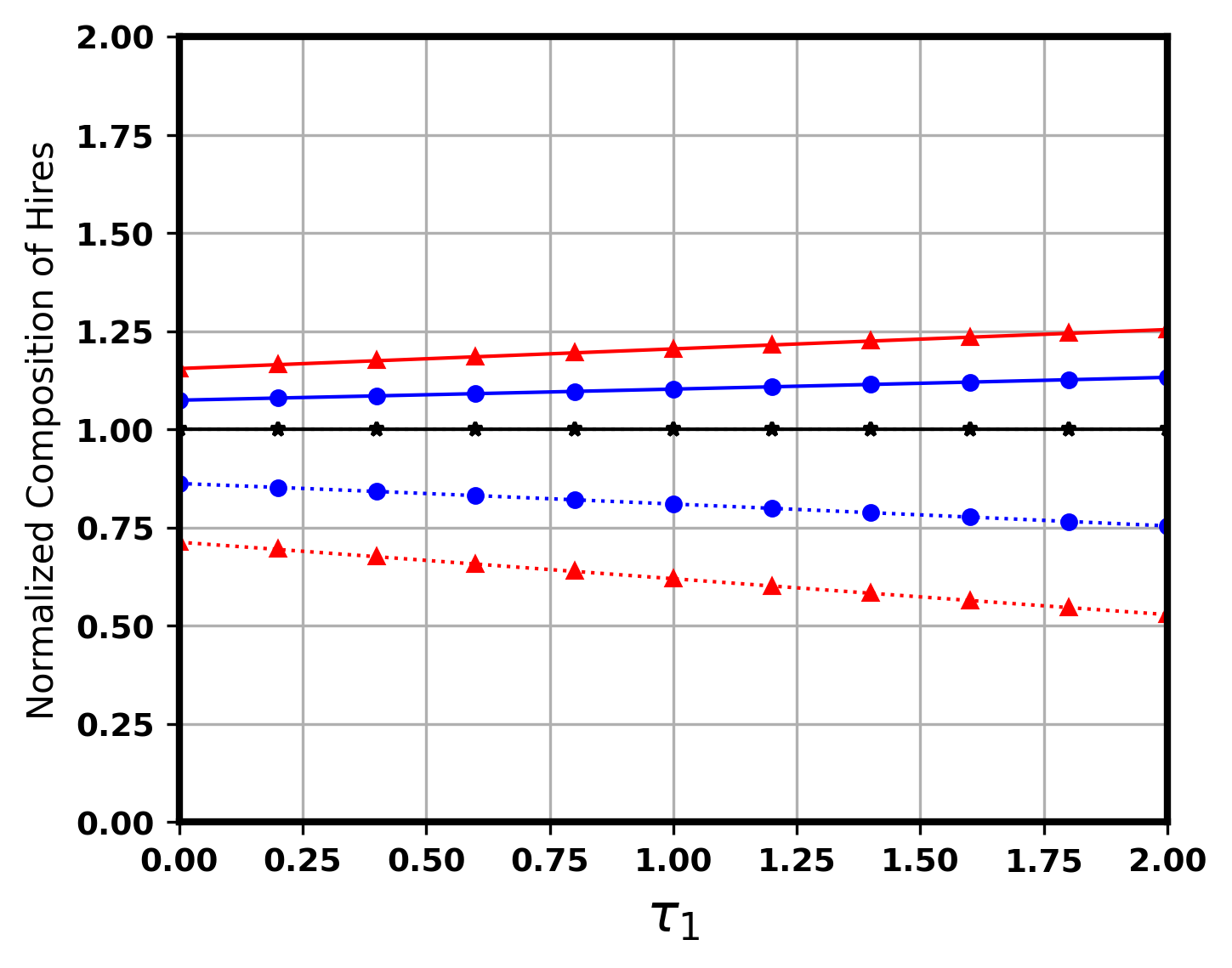}}  
    \subfloat[$\Lambda_A = 0.80$]{\includegraphics[width=0.32\textwidth]{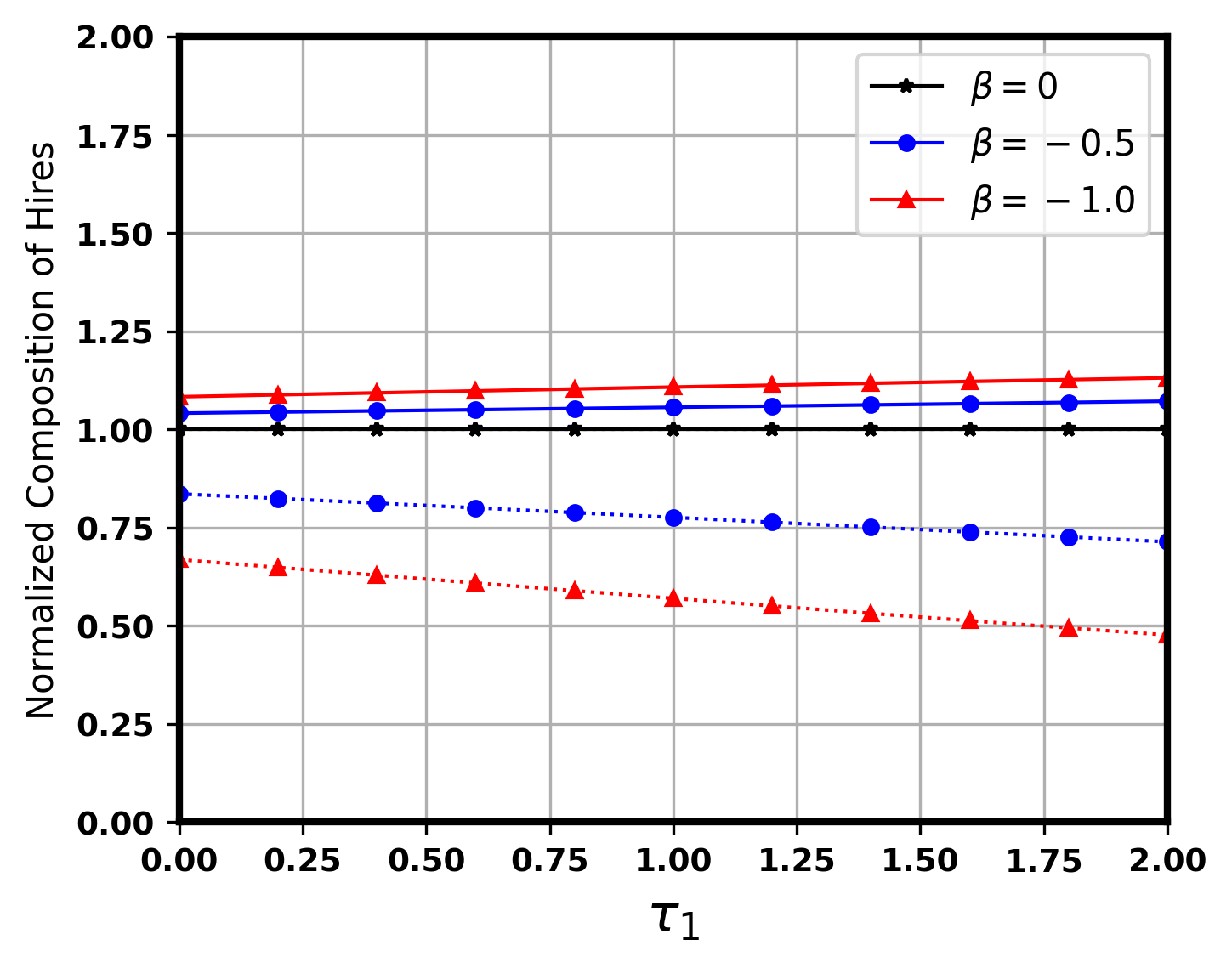}}
    \caption{We plot the expected fraction of hires (normalized by demographic ratios) from group $A$ (solid lines) and group $B$ (dashed lines) respectively as a function of the agent's selection threshold $\tau_1$ for different levels of bias $\beta$ on the mean of group $B$'s observed score distribution ($\tilde s$) and different levels of population skew ($\Lambda_A$). As $\tau_1$ increases, the gap between the groups grows uniformly indicating that the disadvantaged group suffers as the selection process becomes more selective (from blue outwards to red). The same trend is observed in the magnitude of bias ($\beta$). However, as the prevalence of the majority group in the population ($\Lambda_A$) increases, the leading group has little scope for gaining additional advantage, so the extent of disparities actually diminishes.}
    \label{fig:add_bias}
\end{figure}

\begin{figure}[!h]
    \centering
    \subfloat[$\Lambda_A = 0.50$]{\includegraphics[width=0.32\textwidth]{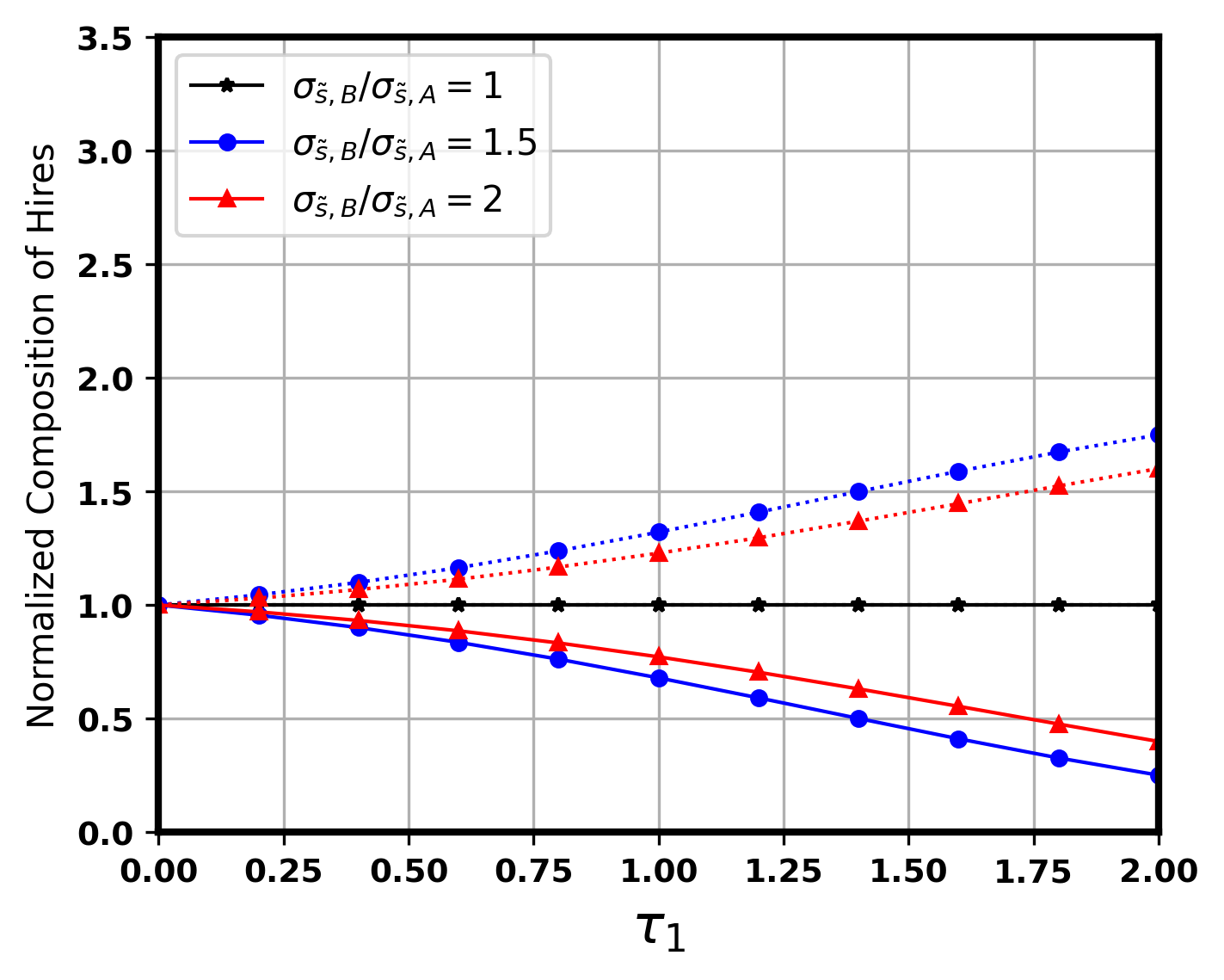}}
    \subfloat[$\Lambda_A = 0.65$]{\includegraphics[width=0.32\textwidth]{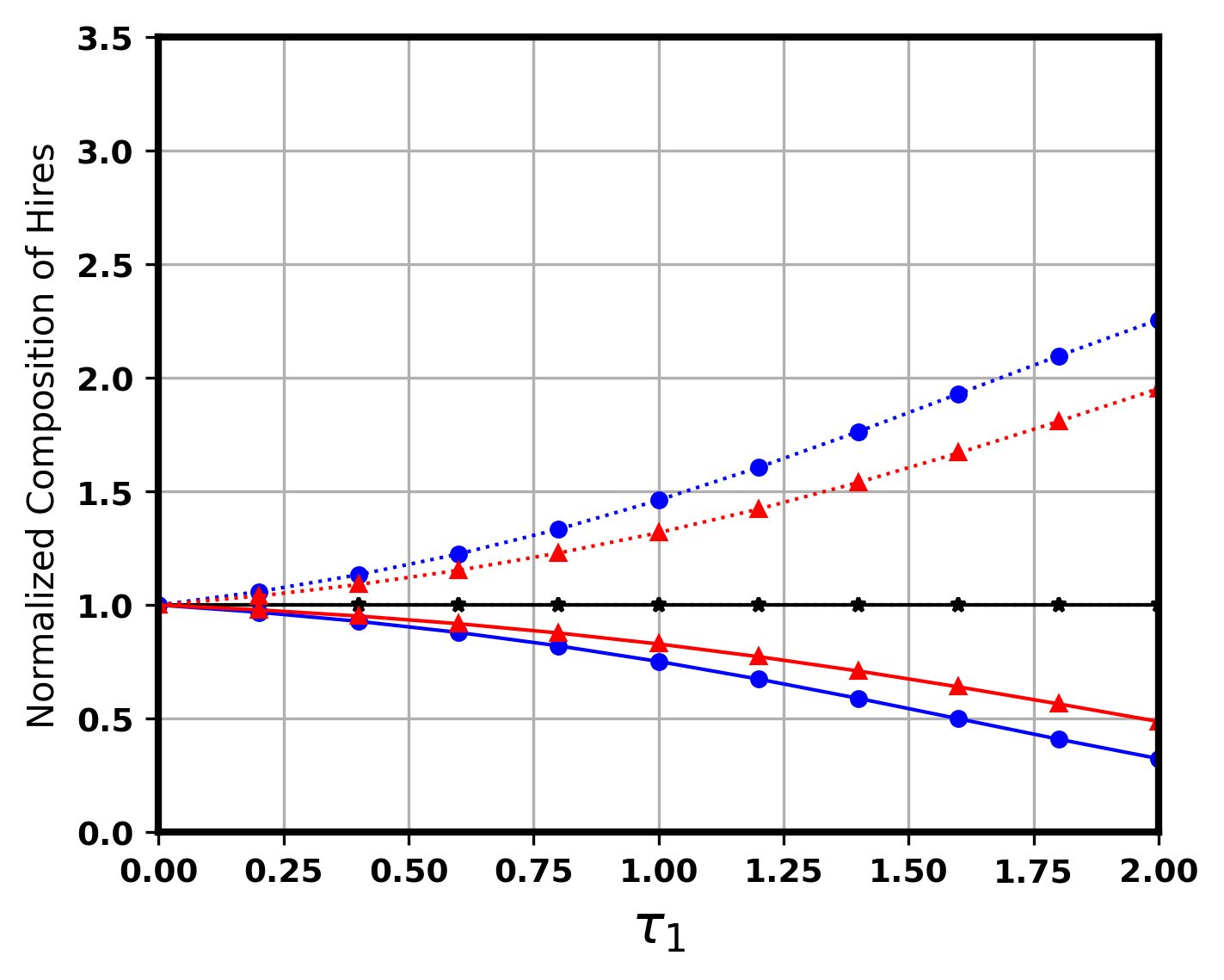}}  
    \subfloat[$\Lambda_A = 0.80$]{\includegraphics[width=0.32\textwidth]{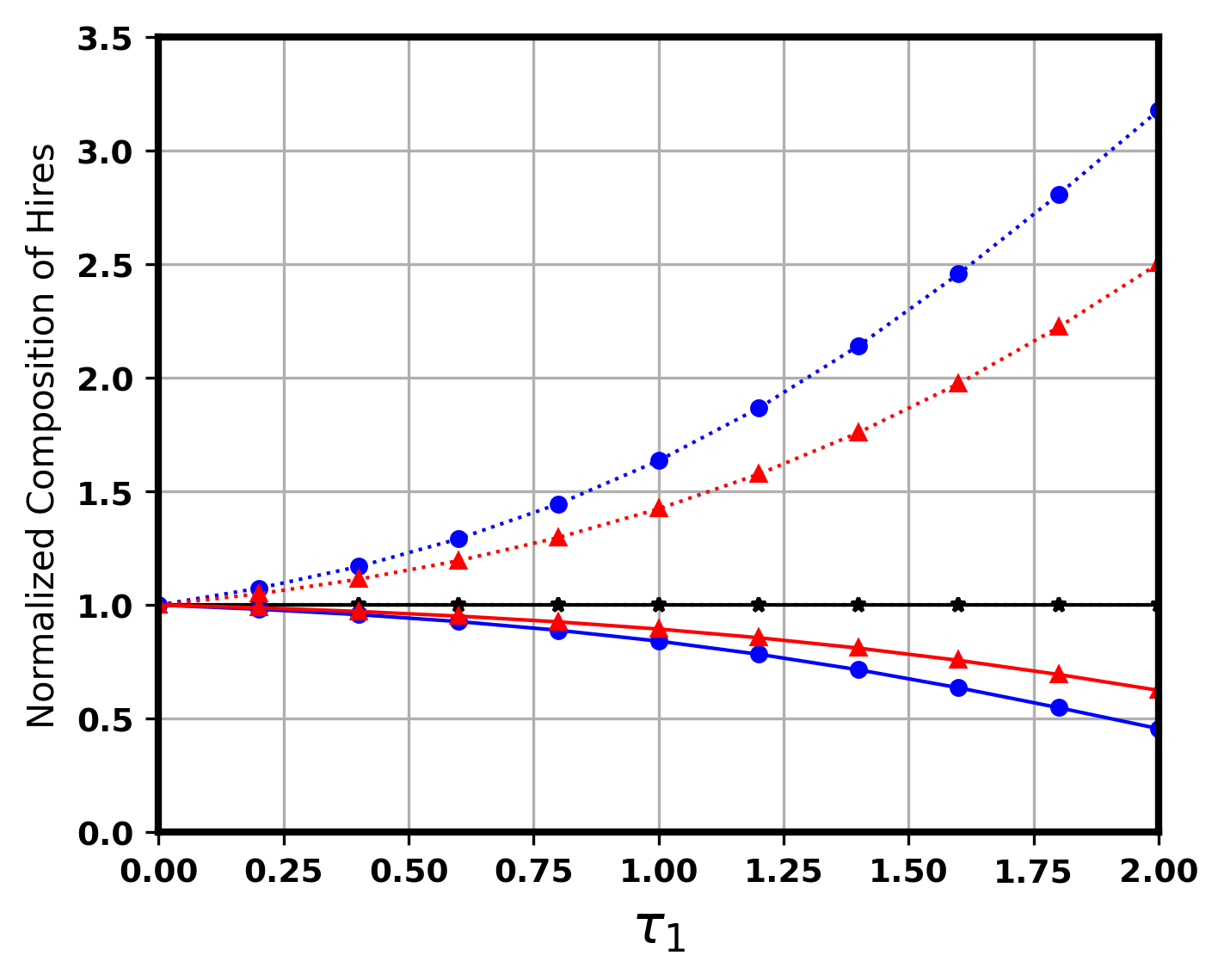}}
    \caption{We plot the expected fraction of hires (normalized by demographic ratios) from group $A$ (solid lines) and group $B$ (dashed lines) respectively as a function of the agent's selection threshold $\tau_1$ ($> 0$) for different levels of the ratio $r = \sigma_{\tilde s, B}/\sigma_{\tilde s, A}$ and different levels of population skew ($\Lambda_A$). The ratio $r$ captures how much more noisy group $B$'s signals are with respect to group $A$. In this case, group $B$ becomes increasingly more favored for selection and unfairness grows in favor of group $B$ as $r$ and $\tau_1$ increase.  }
    \label{fig:var_bias}
\end{figure}

Now let us see what happens in the two special models of disadvantage we introduced earlier: 
\begin{itemize}
    \item \textit{Case I: Negatively biased signal mean ($\beta > 0$, $r = 1$).} In this case, we see that $\mathbb{E}[\mathcal{D}] > 0$, i.e., \textit{the unfairness is in favor of group $A$}. The negatively biased signal mean for group $B$ implies that $\Phi_A(\cdot)$ stochastically dominates $\Phi_B(\cdot)$, so given any threshold $\tau_1$, the probability of any selected applicant to belong to group $A$ is strictly larger. Figure~\ref{fig:add_bias}, where we plot $\mathbb{E}[Y_A/\Lambda_A]$ and $\mathbb{E}[Y_B/\Lambda_B]$ from Equation~\eqref{eq:fairness} as a function of $\tau_1$ for different levels of bias $\beta$ and different levels of population skew $\Lambda_A$, shows this conclusively.  
    \item \textit{Case II: Higher signal variance ($\beta = 0$, $r > 1$).} In this case, the direction of unfairness depends on the sign of $\tau_1$. If $\tau_1 > 0$, $\mathbb{E}[\mathcal{D}] < 0$, i.e., \textit{unfairness actually grows in favor of group $B$}. Due to the higher variance, the distribution of $\tilde s_B$ has fatter tails, which means that the likelihood of a randomly picked applicant, who passes the threshold, to be from group $B$ increases. This can be seen in Figure~\ref{fig:var_bias} (the dashed lines representing group $B$'s normalized hire composition always lie above $1$, indicating that they are hired at rates much higher than what is demographically fair, at the expense of group $A$)\footnote{The sign of $\mathbb{E}[\mathcal{D}]$ does flip when $\tau_1 < 0$, but $\tau_1 < 0$ is generally not a useful case.}. 
\end{itemize}
The above discussion highlights that when the principal delegates to agent and the agent uses a fixed threshold-based selection policy, the selection outcomes can be \textit{unfair}. However, which group benefits from the unfairness is more nuanced and depends on how their signal distributions are related to each other. When it comes to Case II above, it is important to highlight that the disadvantaged group getting favored for selection in expectation is not necessarily a `good outcome'. In fact, we can show that the expected quality of a selected applicant decays monotonically in the level of noise in signal $\tilde s$ (as per Figure~\ref{fig:quality_tlds} in Appendix~\ref{app:extra_fig}), i.e., applicants from group $B$ who get selected are of strictly lower expected quality compared to their group $A$ counterparts. This is because the excessive noise in the signal leads to many `bad' selections from group $B$. This, in turn, can actually reinforce negative stereotypes about group $B$ in future iterations of the selection process, only harming them in the long term; this also creates unfairness \emph{within} group $B$, where less qualified candidates routinely get selected over more qualified ones.

\section{A Selection Process without Delegation}\label{sec:decentralization}

In this section, we consider a selection process where the principal makes selection decisions themselves without delegating to the agent. This is different from the previous setting in the sense that the principal can now decide which applicants to hire albeit at a high cost. In the first part of the section, we characterize how this trade-off plays out and what decisions the principal makes in order to maximize their net utility. Interestingly and surprisingly, we demonstrate how the principal's utility in this setting can actually be non-monotonic in key problem parameters, unlike the setting with delegation. Finally, we conclude the section with an analysis of the multi-group setting, showing that when the principal has the flexibility to set the selection criteria for each group and decide who to hire, the hiring outcomes are completely different from the setting with delegation, leading to significant implications for fairness. 

\subsection{Principal's Optimal Decisions}
Recall that the principal has to choose i) how many applications to review $\nrev$ and ii) what threshold $\tldtau$ on perceived quality ($\tldt$) to use to hire applicants in a way that maximizes their overall utility. The overall utility is given by:
\[
     U_{ndg}(\tldtau, \nrev) = \nrev \cdot \mathbb{P}[\tldt \geq \tldtau] \cdot \mathbb{E}[t~|~\tldt \geq \tldtau] - \nrev \cdot \crev,
\]
where $\crev$ is the cost of reviewing each additional application. Therefore, the principal's optimization is as follows: 
\begin{align}\label{opt:prof}
        \max_{\nrev \geq 0, \tldtau} \quad U_{ndg}(\tldtau, \nrev) \quad \text{s.t.} \quad \nrev \cdot \mathbb{P}[\tldt \geq \tldtau] \leq k. 
\end{align}
Our first result highlights that the solution to optimization program~\eqref{opt:prof} has an important dependence on the cost $\crev$. Further, it shows that the optimal decisions $\nrev^*$ and $\tldtau^*$ can be decoupled and computed efficiently. 

\begin{theorem}\label{thm:opt_n_tau}
The optimal threshold $\tldtau^*$ can be obtained as:
\begin{align}\label{exp:opt_tau}
    \tldtau^* 
    = ~arg\max_{\tldtau} \quad v(\tldtau) \quad \text{where~} v(\tldtau) = \frac{\sigma_t^2}{\sigma_{\tldt}}\cdot \frac{\phi\left(\tldtau/\sigma_{\tldt}\right)}{\Phi^c \left( \tldtau/\sigma_{\tldt} \right)} - \frac{\crev}{\Phi^c\left( \tldtau/\sigma_{\tldt}\right)},
\end{align}
where $\phi(\cdot)$ and $\Phi^c(\cdot)$ denote the PDF and complementary CDF of the standard normal random variable. $\tldtau^*$ is unique and can be computed efficiently. Once we solve for $\tldtau^*$, the optimal $\nrev^*$ can be obtained as follows: 
\begin{itemize}
\item If $\frac{\sigma_{t}^2}{\sigma_{\tldt}}\cdot \frac{1}{\sqrt{2\pi}} > \crev$, then $\nrev^* = \frac{k}{\mathbb{P}[\tldt \geq \tldtau^*]}$;
\item Else, $\nrev^* = 0$.
\end{itemize}
\end{theorem}
The proof for the theorem can be found in Appendix~\ref{pf:thm_opt_tau}. 
A key insight is that the principal's optimal overall utility depends on the cost $\crev$. If the cost is too high, from the principal's point of view there exists no choice of ($\nrev^*, \tldtau^*$) that can provide strictly positive utility. In this case, the only reasonable option is to choose $\nrev^* = 0$. In particular, we have: 

\begin{corollary}\label{corr:viable}
A selection process without delegation is viable\footnote{There are other equivalent versions of Corollary~\ref{corr:viable} that can characterize the viability condition for selection processes without delegation in terms of $\tldtau^*$. For example, we can show that $\frac{\sigma_{t}^2}{\sigma_{\tldt}}\cdot \frac{1}{\sqrt{2\pi}} > \crev \iff \nrev^* > 0 \iff \tldtau^* > 0$. For details, refer to Appendix~\ref{pf:corr_viable}.} for the principal if and only if $\frac{\sigma_{t}^2}{\sigma_{\tldt}}\cdot \frac{1}{\sqrt{2\pi}} > \crev$.
\end{corollary}
For too large costs, the utility earned by the principal is trivial (zero) and it is not viable for them to conduct the selection process on their own.

\subsection{Single Group Setting: A Utilitarian View}
Using the characterization of the optimal decisions in the previous section, we can express the overall expected utility earned by the principal in the selection process without delegation. Note that we are operating in the regime where the cost $\crev$ is small enough ($< \frac{\sigma_{t}^2}{\sigma_{\tldt}}\cdot \frac{1}{\sqrt{2\pi}}$) that such a selection process in viable in the first place.

\begin{lemma}\label{lem:decent_util}
Suppose that $\crev < \frac{\sigma_{t}^2}{\sigma_{\tldt}}\cdot \frac{1}{\sqrt{2\pi}}$ and $\tldtau^*$ is the optimal selection threshold used by the principal. Then the expected utility earned by the principal per hired applicant is equal to $\frac{\sigma_t^2}{\sigma_{\tldt}^2}\cdot \tldtau^*$.
\end{lemma}
The proof can be found in Appendix~\ref{pf:lem_decent_util}. Our main goal here is to understand how the optimal overall expected utility depends on problem parameters. We are particularly interested in the dependence on $\crev$ and $\alpha$. While $\crev$ clearly affects the average quality of hires, it also affects the overall cost. So it is not apriori clear how the net utility might be affected. On the other hand, $\alpha$ captures the degree to which the principal prioritizes applicant fit when measuring quality---as such it is an important intrinsic preference parameter for the principal and needs to be studied. We present two results below capturing these dependencies: 

\begin{claim}\label{clm:util_crev}
In a selection process without delegation which is viable, the optimal overall expected utility for the principal is monotonically decreasing in the cost $\crev$ incurred in reviewing each application. 
\end{claim}
The proof for this claim follows from Lemma~\ref{lem:decent_util} once we make the observation that a higher $\crev$ lowers the optimal selection threshold $\tldtau^*$ (Claim~\ref{clm:tau_n_dep_c} in Appendix~\ref{app:monotonicity}). The intuition is the following: higher reviewing cost means that the principal wants to fill all available positions by reviewing only a small number of applications---this necessitates lowering of the selection standard. 

\begin{claim}\label{clm:util_alpha}
In a selection process without delegation which is viable, the optimal overall expected utility for the principal is \textbf{non-monotonic} in the parameter $\alpha$ that measures their weight on applicant fit. 
\end{claim}
\begin{figure}
    \centering
    \subfloat[$v(\tldtau^*)$ vs $\sigma_t$]{\includegraphics[width=0.3\textwidth]{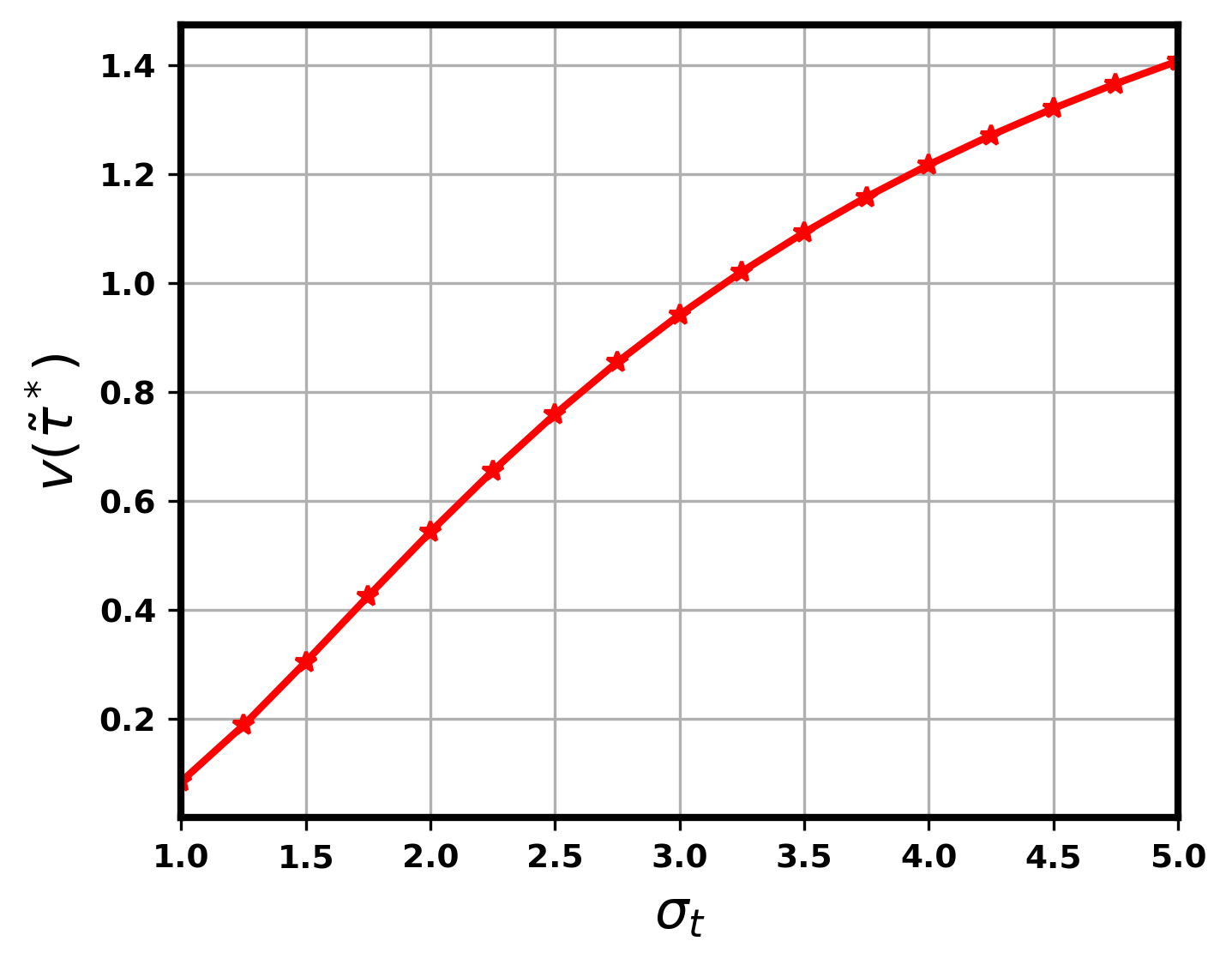}}
    \subfloat[$\sigma_t$ vs $\alpha$]{\includegraphics[width=0.3\textwidth]{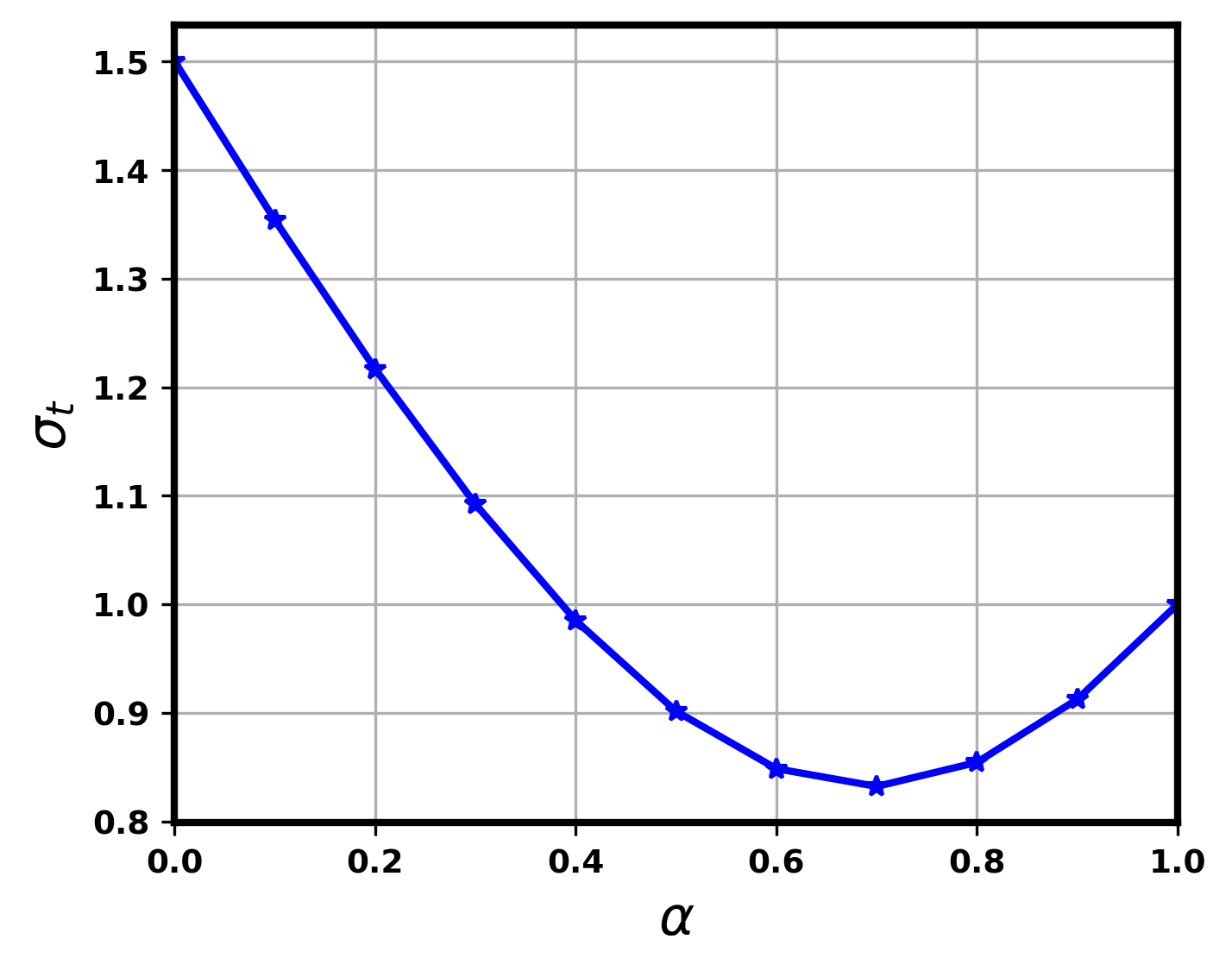}}  
    \subfloat[$v(\tldtau^*)$ vs $\alpha$]{\includegraphics[width=0.3\textwidth]{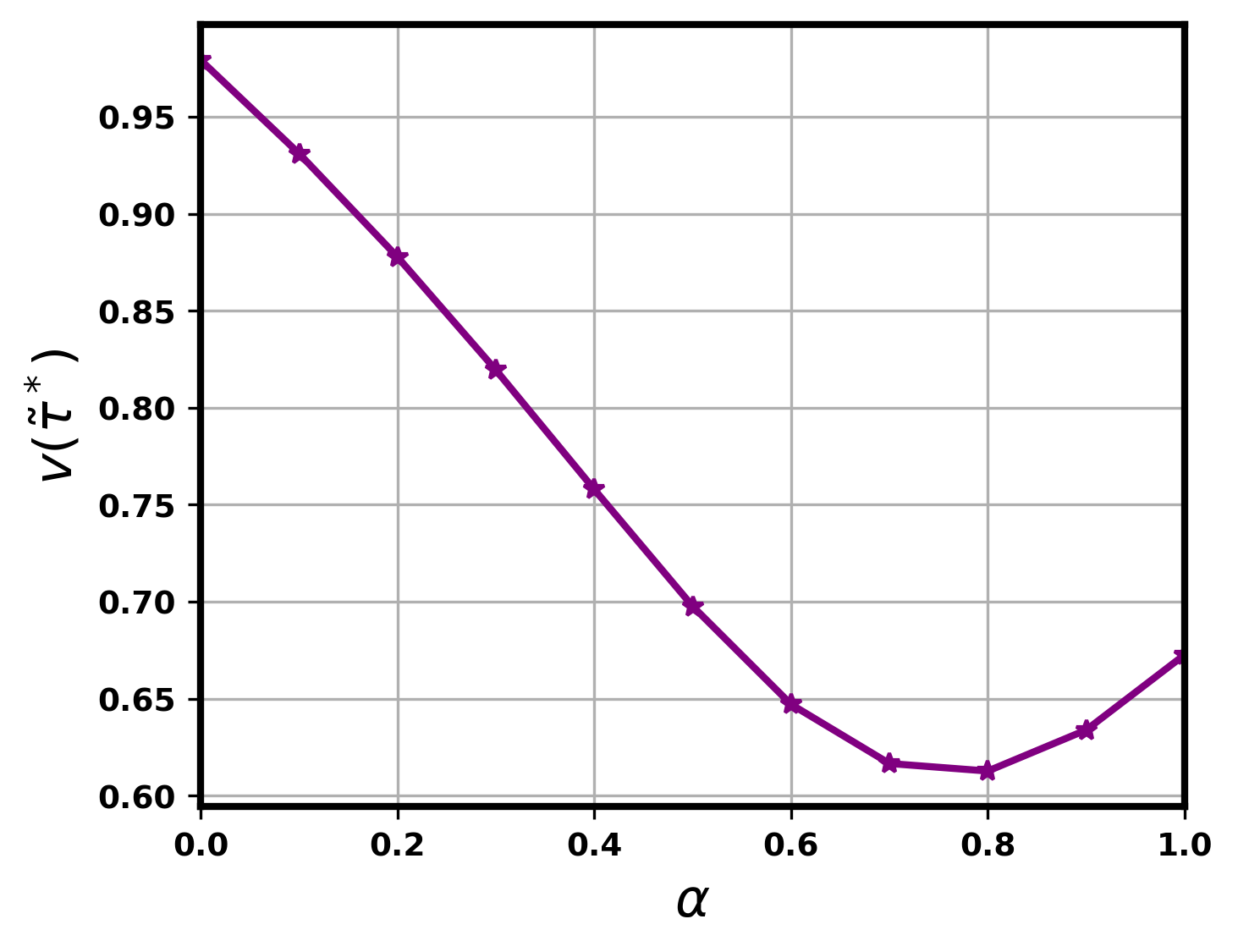}}
    \caption{The principal's net expected utility per hired applicant $v(\tldtau^*)$ is non-monotonic in $\alpha$ when the principal does not delegate. Parameter combinations for sub-figures: (a) $\sigma_e = 2$, $\crev = 0.1$. (b) $\sigma_s = 1.5$, $\sigma_f = 1$. (c) $\sigma_s = 1.5$, $\sigma_f = 1$, $\sigma_{ef} = \sigma_{es} = 0.5$, $\crev = 0.1$. }
    \label{fig:util_vs_alpha}
\end{figure}
Figure~\ref{fig:util_vs_alpha} shows that while the optimal overall expected utility for the principal is monotonically increasing in $\sigma_t$ (a more diverse applicant pool offers better utility on average, follows from Claim~\ref{clm:tau_dep_sigmat} in Appendix~\ref{app:monotonicity}), $\sigma_t$ itself is non-monotonic in $\alpha$. In particular, $\sigma_t$ is $U$-shaped as $\alpha$ increases from $0$ to $1$ (since $\sigma_t = \sqrt{\alpha^2 \sigma_f^2 + (1-\alpha)^2 \sigma_s^2}$ which is convex in $\alpha$ and attains its minimum somewhere in $(0, 1)$). This leads to the overall optimal expected utility to be non-monotonic in $\alpha$. Basically, the intuition is that at intermediate levels of $\alpha$ (when the principal puts approximately equal weights on both applicant fit and ability), finding candidates who are good enough on both metrics becomes harder leading to decreased expected utilities for the principal. Note that this trend is in sharp contrast to the setting with delegation where we showed that principal utilities are monotonically decreasing in $\alpha$.

\subsection{Multi-Group Setting: A Fairness Approach}

Finally, we consider a setting where the applicant pool is mixed and consists of applicants from groups $A$ and $B$. Recall that groups $A$ and $B$ are inherently identical, but Group $B$ is disadvantaged in the sense that either their signals are noisier (with higher variance, i.e., $\sigma_{\tilde s, A} < \sigma_{\tilde s, B}$) or the distribution mean is negatively biased with respect to Group $A$ (bias of $-\beta$). Note that a fundamental difference of this setting with the previous one (with delegation) is that \textit{the principal has complete freedom to set individually optimized selection thresholds for each group in order to maximize utility} and does not suffer from the ``naivety'' of the agent.

We are again interested in the group-wise composition of hires in order to understand the fairness implications of not delegating. To that end, we first show that the disadvantaged group will end up providing strictly worse expected utility per hired applicant for the principal if the disadvantage manifests in the form of higher signal variance. However, if the bias is additive and \emph{fully known} to the principal, it can be accounted for and a principal can simultaneously maximize utility and satisfy demographic parity. .

\begin{lemma}\label{lem:decent_adv_better_util}
Consider a selection process where the principal does not delegate and which is viable. Further, suppose that the principal is allowed to optimize selection thresholds $\tldtau_i^*$ individually for each group $i \in \{A, B\}$, with $v_i(\tldtau_i^*) = \mathbb{E}\left[ t_i~|~\tldt_i \geq \tldtau_i^* \right] - \frac{\crev}{\mathbb{P}\left[ \tldt_i \geq \tldtau_i^*\right]}$ representing the principal's optimal net expected utility for every hired applicant from group $i$. In this case, for $\alpha \in(0,1)$, we have that
\begin{align*}
        &\sigma_{\tilde s, A} = \sigma_{\tilde s, B}~\text{and}~\beta > 0 \implies v_A(\tldtau_A^*) = v_B(\tldtau_B^*) \tag{Biased mean model},\\
         &\sigma_{\tilde s, A} < \sigma_{\tilde s, B}~\text{and}~\beta \geq 0 \implies v_A(\tldtau_A^*) > v_B(\tldtau_B^*). \quad \tag{Disparate variance model}
\end{align*}

Further, in the first case (with negatively biased signal mean for group $B$, but no variance disparity), group $B$'s optimal threshold $\tldtau_B^*$ satisfies $\tldtau_B^* = \tldtau_A^* - (1-\alpha)\beta$. 
\end{lemma}

The proof for the above lemma can be found in Appendix~\ref{pf:lem_decent_adv_better_util}. This result shows that selecting from the group, which is not well-understood, is less lucrative from the principal's point of view ---a high degree of noise leads to inaccurate evaluations and hence poor selection decisions. However, any negative additive biases in the signal mean can be handled as long as the magnitude of the bias can be learnt. We now see what consequences this has on fairness.

Consider the modified optimization problem for the principal when the applicant pool is mixed. Here, the principal is not restricted to selecting a certain number of applicants from either population. Instead, they decide how to simultaneously hire from both populations to fill \textit{shared capacity} $k$.  

\begin{align}\label{opt:joint}
    \max_{\nrev(A), \nrev(B), \tldtau_A, \tldtau_B} \quad &U_{ndg}^{(A)}(\tldtau_A, \nrev(A)) + U_{ndg}^{(B)}(\tldtau_B, \nrev(B)) \quad \text{s.t.}\nonumber\\
    &\nrev(A)\cdot \cP[\tldt_A \geq \tldtau_A] + \nrev(B)\cdot \cP[\tldt_B \geq \tldtau_B] \leq k, \quad \nrev(A), \nrev(B) \geq 0. 
\end{align}
In light of Lemma~\ref{lem:decent_adv_better_util}, we expect the principal to hire more applicants from population $A$ than from population $B$ when group $B$'s signals are noisier (bias only exists in the variance). However, our main result here shows that the reality is even worse, providing dire news for fairness in selection processes without delegation:

\begin{theorem}\label{thm:no_hire_B}
Suppose that group $B$ is disadvantaged because of a noisier signal $\tilde s$ (biased signal variance, but no bias in the signal mean) and the cost $\crev$ is low enough that the selection process without delegation is viable. Then, at the optimal solution to Problem~\ref{opt:joint}, $\nrev(B)^* = 0$, i.e., no applicant from Group $B$ is hired. 
\end{theorem}
The proof for the theorem can be found in Appendix~\ref{app:no_hire_B}. 
Returning to our graduate admissions example, PhD applicants from foreign or lesser-known educational backgrounds often face a disadvantage because their transcripts/GPA/letters may not be interpretable in the same way as their peers from well-known institutions. When faced with such 
noisy signals, faculty often choose to err on the side of caution, hiring students overwhelmingly from backgrounds they are familiar with. We see the same effects at play in Theorem~\ref{thm:no_hire_B}, where \textit{the principal pre-excludes group $B$ and then selects applicants from the other group using threshold $\tldtau_A^*$}.    

The outcome, however, is completely different for the setting where a negative bias exists in group $B$'s signal mean for $\tilde s$. Since additive biases can be corrected for, no group disparities arise: 
\begin{theorem}\label{thm:no_disparity}
Suppose that group $B$ is disadvantaged because of a negatively biased mean for signal $\tilde s$ (but no bias in signal variance) and the cost $\crev$ is low enough that the selection process without delegation is viable. Then, there exists an optimal solution to Problem~\ref{opt:joint} where the group-wise composition of hires satisfies demographic parity.
\end{theorem}
This again follows directly from Lemma~\ref{lem:decent_adv_better_util}. When the bias $\beta$ is known, the principal can correct for the bias by augmenting the signal values $\tilde s$ for all group $B$ applicants by amount $\beta$. As a result, the distribution of quality for both groups become virtually indistinguishable, thereby leading to the above outcome. %

\section{To Delegate or Not: Efficiency \& Fairness Implications}\label{sec:comparison}

We conclude this paper by conducting a comparative analysis between the two selection processes discussed so far and investigating when it can be beneficial for a principal to delegate to an agent. We focus on three primary dimensions of comparison: i) the utility of the principal, ii) the average quality of selected applicants from the principal's perspective, and iii) the fairness of hiring outcomes. While the first two pertain to the quality and efficiency of the hiring process, ensuring that hiring processes are \textit{fair} is also of paramount importance. This is particularly crucial in light of many real-world instances of biased or discriminatory hiring practices \citep{amazon,kline2022systemic}.

\paragraph{Principal Utilities \& Expected Quality of Selected Applicants.}

First, we explore conditions under which it is beneficial for the principal to delegate to the agent, if the goal is to obtain i) higher expected quality for selected applicants; or ii) higher expected utilities per selected applicant. 
Figure~\ref{fig:util_comparison} shows the difference in net principal utilities per selected applicant ($\Delta_{utility}$) and the difference in expected quality of a selected applicant ($\Delta_{quality}$) between the settings with and without delegation, as a function of key problem parameters. Formally, we define:
\[
      \Delta_{quality} = \mathbb{E}\left[t~|~\tilde s \geq \tau_1 \right] - \mathbb{E}\left[t ~|~ \tilde t \geq \tldtau^* \right], \quad \text{and} 
\]
\[
      \Delta_{utility} = \left(\mathbb{E}\left[t~|~\tilde s \geq \tau_1 \right]\right) - \left( \mathbb{E}\left[t ~|~ \tilde t \geq \tldtau^* \right] - \frac{\crev}{\mathbb{P}\left[\tldt \geq \tldtau^* \right]} \right). 
\]
A positive sign for $\Delta_{quality}$ and $\Delta_{utility}$ identifies preferences towards delegation. We make the following observations:

(1) \textit{Dependence on $\tau_1$.} At any fixed $\alpha \in (0, 1)$, the expected quality of applicants selected by the agent increases monotonically in the selection threshold $\tau_1$, as shown in Corollary~\ref{cor:tau1_monotone}; at the same time, the choice of $\tau_1$ does not affect the case where the principal does not delegate. This implies that $\Delta_{quality}$ and $\Delta_{utility}$ are both \emph{monotonically increasing} in $\tau_1$. Therefore, for any given $\alpha$, there exists a $\tau_1(\alpha)$ such that delegation is strictly better for the principal in terms of both expected quality of hire and net expected utility per hire for all $\tau_1 \geq \tau_1(\alpha)$. 

(2) \textit{Dependence on $\alpha$.} At fixed $\tau_1$, whether it is beneficial to delegate or not can be \textbf{non-monotonic} in $\alpha$. This can be seen in Figure~\ref{fig:util_comparison}---note that $\Delta_{utility}$ and $\Delta_{quality}$ are both inverted $U$-shaped and there exists intermediate regimes of $\alpha$ where delegation would be beneficial (see subfigures b) and c)). Intuitively, at intermediate values of $\alpha$, finding a good candidate may be costly for the principal (the candidate needs to be good along both dimensions of ability and fit and hence more applications need to be reviewed), and delegation to the agent can be beneficial. The agent can compensate for applicants doing poorly on fit by ensuring that they do exceptionally well on ability. 

However, given any $\tau_1$, there always exists some $\alpha(\tau_1) \in (0, 1)$ such that for $\alpha \geq \alpha (\tau_1)$, delegating is never beneficial. Intuitively, once the principal and agent preferences become sufficiently misaligned, it is always better for the principal to retain control over decision-making. As $\tau_1$ becomes smaller, this $\alpha(\tau_1)$ also becomes smaller and delegation eventually may become non-beneficial for all $\alpha$ (for example, see subfigure~\ref{fig:util_comparison}a). This is the regime where the average ability of applicants selected by the agents is not sufficient to compensate for the lack of selection based on fit. 

For a more detailed overview of the delicate balance between $\tau_1$ and $\alpha$, we refer the reader to Figure~\ref{fig:heatmaps} in Appendix~\ref{app:extra_fig}.

\begin{figure}
    \centering
    \subfloat[$\tau_1 = 0.5$]{\includegraphics[width=0.3\textwidth]{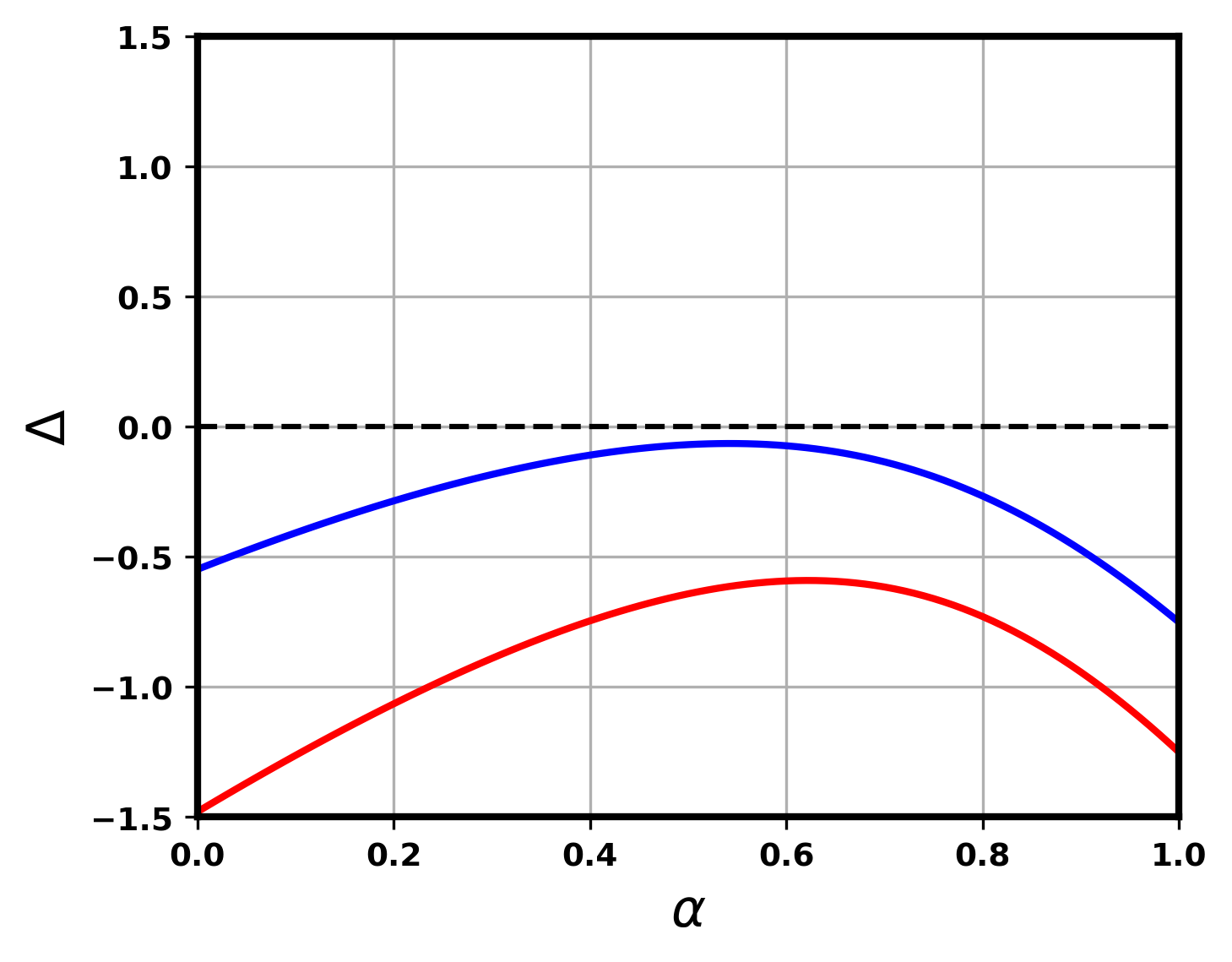}} 
    \subfloat[$\tau_1 = 1.0$]{\includegraphics[width=0.3\textwidth]{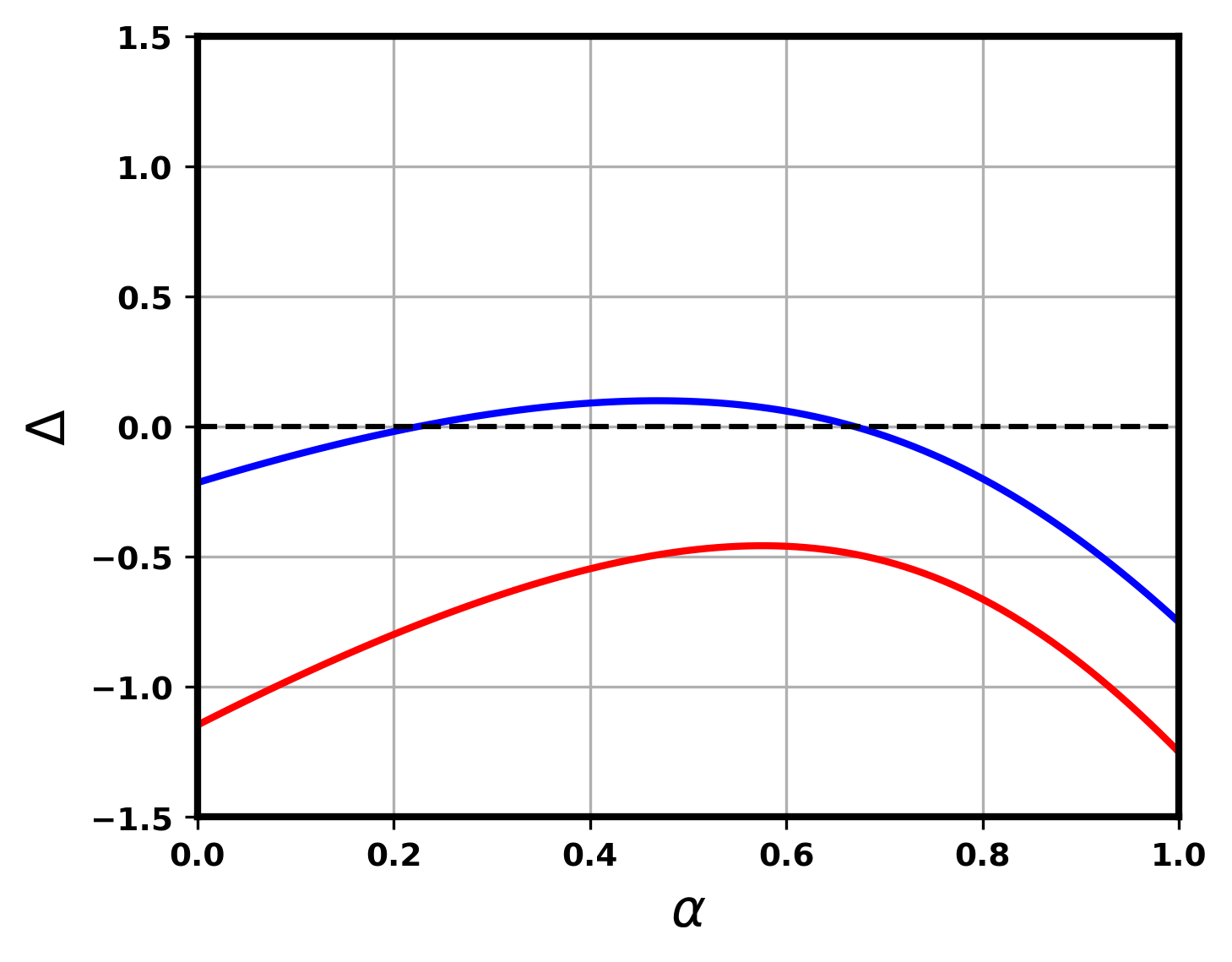}} 
    \subfloat[$\tau_1 = 2.4$]{\includegraphics[width=0.3\textwidth]{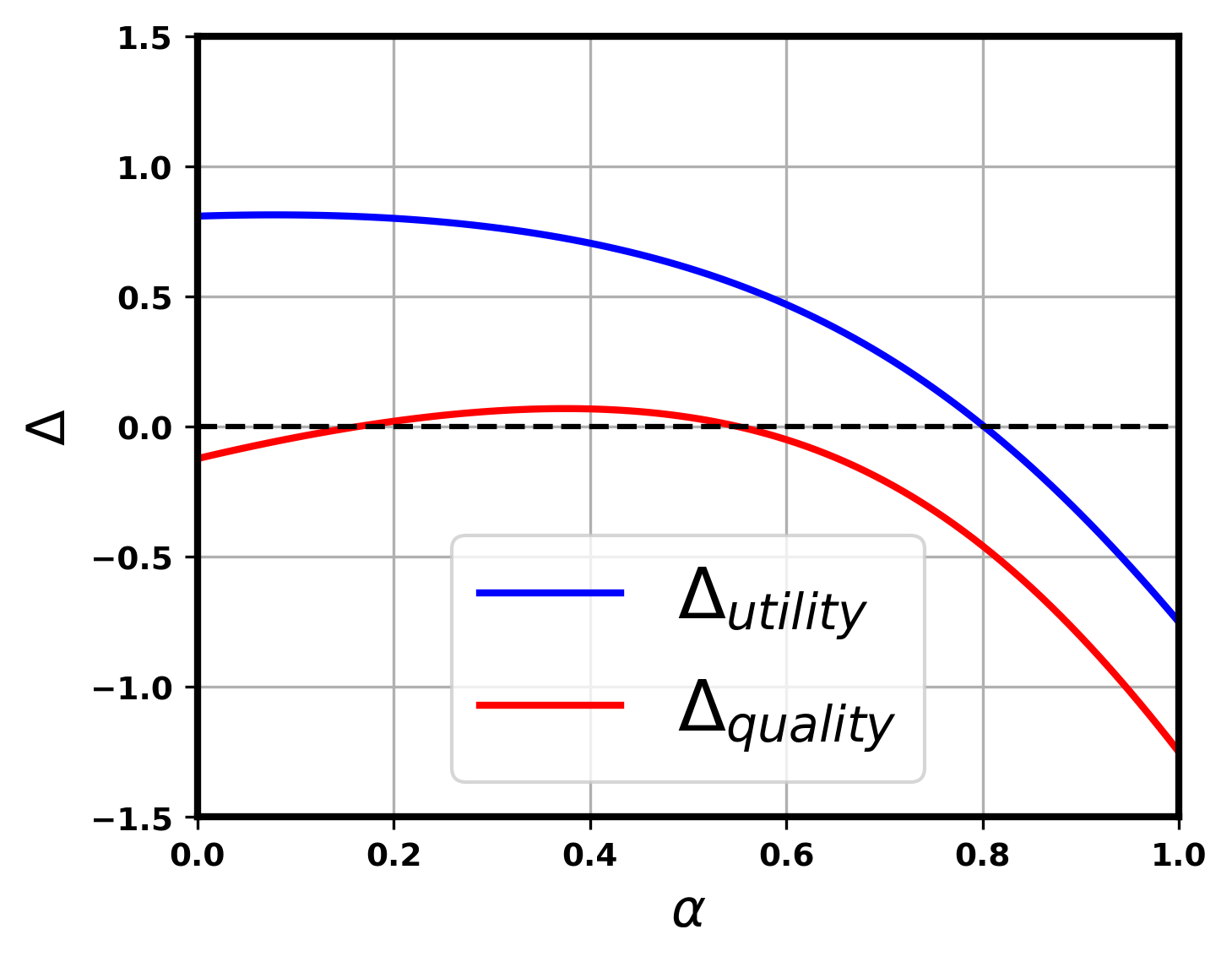}}
    \caption{We plot $\Delta_{quality}$ and $\Delta_{utility}$ as a function of $\alpha$ for different levels of agent's threshold $\tau_1$.  Positive values of $\Delta$ indicate that the principal prefers delegation. Parameter combination: $\crev = 0.1$, $\sigma_f = 1$, $\sigma_s = 2$, $\sigma_{\tilde f} = 1.12$, $\sigma_{\tilde s} = 2.06$. We provide a more comprehensive graphical representation of regimes where delegation is beneficial in Appendix~\ref{app:extra_fig}.}
    \label{fig:util_comparison}
\end{figure}

\paragraph{Fairness in Multi-Group Settings.} Building up on Sections~\ref{sec:centralization} and~\ref{sec:decentralization}, we note that neither process can achieve \textit{fairness} (in terms of demographic parity) across all settings. 
However, the modes of failure differ across the two processes. When a \emph{naive} agent makes decisions on behalf of the principal, it applies a group-blind selection standard. On the other hand, when no delegation occurs, not only does the principal have access to more information about each applicant, they are also strategic and utility-maximizing. This fundamental asymmetry explains why the fairness outcomes in the latter setting are markedly different from the former.

Importantly, \textit{the nature of disadvantage faced by a group directly determines whether unfairness caused by this disadvantage can be mitigated or not.} For example, when the disadvantage exists purely as additive biases in the signal mean, a decision-maker (like the principal) can correct for it through group-aware selection policies, without incurring any dis-utility; in this case, not delegating leads to significantly better outcomes for fairness. However, there are some other forms of disadvantage that cannot be corrected for without sacrificing process efficiency. For example, in the disparate variance setting, a rational principal would choose to completely exclude the disadvantaged group (because signals are unreliable) to minimize the risk of hiring poor quality applicants and lowering their utility. In this case, neither of our selection processes is capable of achieving fair outcomes, but delegating to the agent at least ensures selections from both groups (albeit at unfair rates). On a broader note, this highlights the importance of i) understanding the exact type of disparities faced by different populations; and ii) the extent of flexibility a decision maker has while designing selection processes. There is no one-size-fits-all solution.

\section*{Acknowledgement}
The authors are grateful for support from the US National Science Foundation (NSF) under
grants IIS-2504990 and IIS-2336236. Any opinions and findings expressed in this
material are those of the authors and do not reflect the views of their funding agencies.

\bibliography{mybib.bib}
\bibliographystyle{plainnat}

\appendix

\section{Proofs for Section~\ref{sec:centralization}}

\subsection{Proof of Lemma~\ref{lem:tech_prof_util_cent}}\label{pf:lemma1}
Recall that the principal's expected utility from a applicant hired through a selection process which has been delegated to the agent is given by $\mathbb{E}\left[t~|~\tilde s \geq \tau_1 \right]$ 
Before we proceed with the main proof, we need to introduce some technical results:
\begin{claim}\label{tech_clm:cond_dist}
If two random variables $X$ and $Y$ are such that $X \sim \mathcal{N}(\mu_X, \sigma_X^2)$, $Y \sim \mathcal{N}(\mu_Y, \sigma_Y^2)$ and $\Sigma$ is the covariance between $X$ and $Y$, then the conditional distribution of $X$ conditional on $Y = y$ is given by: 
\[
     X~|~(Y = y) \sim \mathcal{N}\left(\mu_X + \frac{\Sigma}{\sigma_Y^2}\cdot \left(y - \mu_Y\right), \sigma_X^2 - \frac{\Sigma^2}{\sigma_Y^2} \right)
\]
\end{claim}
\begin{proof}
For a detailed proof of the above claim, please refer to these notes \cite{chuongnotes}.
\end{proof}
\noindent
Noting that $s \sim \mathcal{N}(0, \sigma_s^2)$, $\tilde s \sim \mathcal{N}(0, \sigma_{\tilde s}^2)$ and $Cov(s, \tilde s) = \sigma_s^2$, we can use the above claim to conclude that: 
\[
    \mathbb{E}[s~|~\tilde s = x] = \frac{\sigma_s^2}{\sigma_{\tilde s}^2}x.
\]
We will now use this result to prove the following technical result that links $\mathbb{E}[s~|~\tilde s \geq \tau_1]$ to $\mathbb{E}[\tilde s~|~\tilde s \geq \tau_1]$: 
\begin{claim}\label{tech_clm:intermediate}
We must have: 
    $\mathbb{E}[s~|~\tilde s \geq \tau_1] = \left(\frac{\sigma_s}{\sigma_{\tilde s}}\right)^2 \cdot \mathbb{E}[\tilde s~|~\tilde s \geq \tau_1]$.  
\end{claim}
\begin{proof}
The proof follows directly from Claim~\ref{tech_clm:cond_dist} and the law of total expectation. 
\begin{align*}
    \mathbb{E}[s~|~\tilde s \geq \tau_1] &= \mathbb{E}\left[ \mathbb{E}[s~|~\tilde s]~|~ \tilde s \geq \tau_1 \right] \quad (\text{law of total expectation}) \\
    &= \mathbb{E}\left[ \frac{\sigma_s^2}{\sigma_{\tilde s}^2}\tilde s ~|~\tilde s \geq \tau_1 \right] \quad (\text{by Claim \ref{tech_clm:cond_dist}})\\
    &= \frac{\sigma_s^2}{\sigma_{\tilde s}^2} \cdot \mathbb{E}[\tilde s~|~\tilde s \geq \tau_1].
\end{align*}
\end{proof}
Now, if we can compute $\mathbb{E}[\tilde s~|~\tilde s \geq \tau_1]$, we are done. To this end, we present technical Claim~\ref{tech_clm:hazard} (without proof) which provides the closed form expression for the mean of a truncated normal random variable. 
\begin{claim}\label{tech_clm:hazard}
Let $X \sim N(\mu_X, \sigma_X^2)$. Suppose, the lower tail of $X$ is truncated at $a$. Then, 
\[
       \mathbb{E}[X~|~X \geq a] = \mu_X + \sigma_X \cdot H\left(\frac{a-\mu_X}{\sigma_X}\right),
\]
where $H(y) = \frac{\phi(y)}{1-\Phi(y)}$ represents the hazard rate of a standard normal random variable. 
\end{claim}
\begin{proof}
For completeness, we provide a detailed proof in Appendix~\ref{app:hazard}.
\end{proof}
Noting that $\tilde s \sim \mathcal{N}(0, \sigma_{\tilde s}^2)$ and using Claim~\ref{tech_clm:hazard}, we have $\mathbb{E}[\tilde s~|~\tilde s \geq \tau_1] = \sigma_{\tilde s}\cdot H\left(\frac{\tau_1}{\sigma_{\tilde s}} \right)$. 
Putting everything together, we have:
\begin{align*}
    \E[t~|~\tilde s \geq \tau_1] &= \E[\alpha f + (1-\alpha)s~|~\tilde s \geq \tau_1] \\
    &= \alpha \cdot \E[f~|~\tilde s \geq \tau_1] + (1-\alpha) \cdot \E[s~|~\tilde s \geq \tau_1] \\
    &= \alpha \cdot \E[f] + (1-\alpha) \cdot \E[s~|~\tilde s \geq \tau_1] \quad \text{(since $f\perp s, f\perp \epsilon_s$)}\\
    &= (1-\alpha) \cdot \E[s~|~\tilde s \geq \tau_1] \quad \text{(since $\E[f] = 0$)}\\
    &= (1-\alpha)\cdot \frac{\sigma_s^2}{\sigma_{\tilde s}^2}\cdot \E[\tilde s ~|~\tilde s \geq \tau_1] \quad (\text{Claim~\ref{tech_clm:intermediate}})\\
    &= (1-\alpha)\cdot \frac{\sigma_s^2}{\sigma_{\tilde s}^2}\cdot \sigma_{\tilde s} \cdot H\left(\frac{\tau_1}{\sigma_{\tilde s}} \right) \quad (\text{Claim~\ref{tech_clm:hazard}})\\
    &= \frac{(1-\alpha)\sigma_s^2}{\sigma_{\tilde s}}\cdot H\left(\frac{\tau_1}{\sigma_{\tilde s}}\right).
\end{align*}

\subsection{Proof of Lemma~\ref{lem:cent_comp}}\label{pf:lem2}

Suppose, groups $A$ and $B$ have demographic ratios of $\lambda$ and $1-\lambda$ respectively. Let $\Phi_A(\cdot)$ and $\Phi_B(\cdot)$ indicate the CDFs of the test score distributions of each group. The test score distribution $\Phi_M(\cdot)$ of the mixture distribution is therefore given as follows: 
\[
     \Phi_M(x) = \lambda \cdot \Phi_A(x) + (1-\lambda)\cdot \Phi_B(x) \quad \forall~x \in \mathbb{R},
\]
and the corresponding complementary CDF is given by: 
\[
    \bar \Phi_M(x) = \lambda \cdot \bar \Phi_A(x) + (1-\lambda)\cdot \bar \Phi_B(x) \quad \forall~x \in \mathbb{R}
\]
Now, 
\begin{align*}
    &\mathbb{P}\left[ \text{randomly picked applicant belongs to Gr A} ~|~ \text{test score}\geq \tau_1 \right] \\
    = & \frac{\mathbb{P}\left[ \text{randomly picked applicant belongs to Gr A, test score $\geq \tau_1$} \right] }{ \mathbb{P}\left[ \text{a random draw from $\Phi(\cdot)$ has a value $\geq \tau_1$} \right] } \\
    = & \frac{ \mathbb{P}\left[ \text{randomly picked applicant belongs to Gr A} \right] \cdot \mathbb{P}\left[ \text{test score $\geq \tau_1$} ~|~ \text{applicant belongs to Gr A} \right]}{ \mathbb{P}\left[ \text{a random draw from $\Phi_M(\cdot)$ has a value $\geq \tau_1$} \right] } \\
    = & \frac{\lambda \cdot \bar \Phi_A(\tau_1)}{ \bar \Phi_M(\tau_1) }. 
\end{align*}
Similarly, for group $B$, 
\[
      \mathbb{P}\left[ \text{randomly picked applicant belongs to Gr B} ~|~ \text{test score}\geq \tau_1 \right] = \frac{(1-\lambda) \cdot \bar \Phi_B(\tau_1)}{ \bar \Phi_M(\tau_1) }.
\]
This concludes the proof of the lemma.

\subsection{Proof of Theorem~\ref{thm:D_stats}}\label{pf:thm1}
Suppose that $K$ applicants have been hired from the mixed applicant pool at threshold $\tau_1$ through a selection process delegated to the agent. Let $X_i$ be the indicator random variable that takes value $1$ if hired applicant $i$ belongs to Group $A$, $0$ otherwise. Lemma~\ref{lem:cent_comp} shows that $X_i \sim \text{Ber}\left( \frac{ \Lambda_A \cdot \bar \Phi_A(\tau_1)  }{\bar \Phi_M(\tau_1)} \right)$. Note that can express the fairness metric $\mathcal{D}$ in terms of $X_i$'s as follows: 
\begin{align*}
    \mathcal{D} = \frac{1}{K}\sum_{i=1}^K \underbrace{\left(\frac{X_i}{\Lambda_A} - \frac{1-X_i}{\Lambda_B} \right)}_{Z_i}.
\end{align*}
Therefore, for $i \in [K]$, $Z_i$ has the following pmf: 
\[
         Z_i = \begin{cases}
            & \frac{1}{\lambda} \quad \text{w.p.} \quad \frac{ \lambda \cdot \bar \Phi_A(\tau_1)  }{\bar \Phi_M(\tau_1)}, \\
            & - \frac{1}{1-\lambda} \quad \text{w.p.} \quad \frac{ (1-\lambda) \cdot \bar \Phi_B(\tau_1)  }{\bar \Phi_M(\tau_1)}.
         \end{cases}
\] 
Then, we have $\mathbb{E}[Z_i] = \frac{ \bar \Phi_A(\tau_1) - \bar \Phi_B(\tau_1) }{\bar \Phi_M(\tau_1)}$. 
Further note that $Z_i$'s are mutually independent. Now, 
\begin{align*}
    \mathbb{E}[\vert\mathcal{D}\vert] = \mathbb{E}\left[ \frac{1}{K} \left\vert \sum_{i=1}^K Z_i \right\vert \right] 
        = \frac{1}{K}\mathbb{E} \left[ \left\vert \sum_{i=1}^K Z_i \right\vert \right] \leq \frac{1}{K} \mathbb{E}\left[ \sum_{i=1}^K |Z_i|  \right] = \mathbb{E}[|Z_i|] = \frac{ \bar \Phi_A(\tau_1) + \bar \Phi_B(\tau_1) }{\bar \Phi_M(\tau_1)}. 
\end{align*}
The inequality used above follows from triangle inequality. We can also derive a lower bound using Jensen' inequality for convex functions as follows: 
\begin{align*}
    \mathbb{E}[\vert\mathcal{D}\vert] = \mathbb{E}\left[ \frac{1}{K} \left\vert \sum_{i=1}^K Z_i \right\vert \right] 
        = \frac{1}{K}\mathbb{E} \left[ \left\vert \sum_{i=1}^K Z_i \right\vert \right] \geq \frac{1}{K} \left\vert \mathbb{E}[ \sum_{i=1}^K Z_i] \right\vert = \left\vert \mathbb{E}[Z_i] \right\vert = \frac{ \vert \bar \Phi_A(\tau_1) - \bar \Phi_B(\tau_1) \vert }{\bar \Phi_M(\tau_1)}. 
\end{align*}


\section{Proofs for Section~\ref{sec:decentralization}}

\subsection{Proof of Theorem~\ref{thm:opt_n_tau}}\label{pf:thm_opt_tau}
We will complete the proof in $3$ parts: in the first part, we will compute the conditional expectation term so that we can rewrite the objective function in problem~\eqref{opt:prof} in closed form. In the second part, we will argue how to solve the optimization problem itself. Finally, in part $3$, we will argue about properties of the optimal threshold that enable efficient computation. 

\subsubsection*{Part $1$: Computing $\E[t~|~\tldt \geq \tldtau]$}
\begin{align*}
    \E[t~|~\tldt \geq \tldtau] &= \frac{\sigma_t^2}{\sigma_{\tldt}^2}\cdot \E[\tldt ~|~\tldt \geq \tldtau] \quad (\text{identical to Claim~\ref{tech_clm:intermediate}})\\
    &= \frac{\sigma_t^2}{\sigma_{\tldt}^2}\cdot \sigma_{\tldt} \cdot H\left(\frac{\tldtau}{\sigma_{\tldt}} \right) \quad (\text{identical to Claim~\ref{tech_clm:hazard}})\\
    &= \frac{\sigma_t^2}{\sigma_{\tldt}}\cdot \frac{ \phi \left(\tldtau/\sigma_{\tldt}\right) }{\Phi^c \left(\tldtau/\sigma_{\tldt}\right)  }.
\end{align*}

\subsubsection*{Part $2$: Solving the optimization problem}
Now, we define: 
\begin{align*}
    O^* &= \max_{\nrev \geq 0, \tldtau} \quad \nrev \cdot \cP[\tldt \geq \tldtau]\cdot \E[t~|~\tldt \geq \tldtau] - \nrev\cdot \crev \\
    &s.t. \quad \nrev \cdot \cP[\tldt\geq \tldtau] \leq k. 
\end{align*}
Now, note that any solution $(n, \tldtau)$ of the form $(0, \tldtau)$ is feasible and produces an objective value of $0$. Therefore, $O^* \geq 0$. There are $2$ cases: 
\begin{enumerate}
    \item Case $1$ ($O^* > 0$): Suppose, $(\nrev^*, \tldtau^*)$ is the optimal solution which produces $O^*$. Clearly, $\nrev^* > 0$ and $\cP[\tldt \geq \tldtau^*]\cdot \E[t~|~\tldt \geq \tldtau^*] - \crev > 0$. We claim that the constraint: $\nrev \cdot \cP[\tldt \geq \tldtau] \leq k$ must be active at ($\nrev^*, \tldtau^*$). This follows directly from the fact that the objective function is increasing in $\nrev$ at the optimal solution $\tldtau^*$. Therefore, we must have: 
    \[
         \nrev^* = \frac{k}{\cP[\tldt \geq \tldtau^*]} \quad \text{if~} \cP[\tldt \geq \tldtau^*]\cdot \E[t~|~\tldt \geq \tldtau^*] > \crev. 
    \]
    In this case, we can reduce the optimization problem to the following form: 
    \[
         \max_{\tldtau} \quad \frac{k}{\cP[\tldt \geq \tldtau]} \cdot \cP[\tldt \geq \tldtau]\cdot \E[t~|~\tldt \geq \tldtau] - \crev \cdot \frac{k}{\cP[\tldt \geq \tldtau]} 
    \]
    or equivalently, 
    \[
      \tldtau^* = arg\max_{\tldtau} \quad \E[t~|~\tldt \geq \tldtau] - \frac{\crev}{\cP[\tldt \geq \tldtau]}
    \]
    Substituting the expressions for $\E[t~|~\tldt \geq \tldtau]$ and $\cP[\tldt \geq \tldtau]$, we obtain: 
    \[
        \tldtau^* = ~arg\max_{\tldtau} \quad \underbrace{\frac{\sigma_t^2}{\sigma_{\tldt}}\cdot \frac{\phi\left(\tldtau/\sigma_{\tldt}\right)}{\Phi^c\left(\tldtau/\sigma_{\tldt}\right)} - \frac{\crev}{\Phi^c\left(\tldtau/\sigma_{\tldt}\right)}}_{v(\tldtau)}
    \]
    \item Case $2$ ($O^* = 0$): From the analysis in Case $1$, it is clear that $O^* = 0$ if and only if $\cP[\tldt \geq \tldtau]\cdot \E[t~|~\tldt \geq \tldtau] - \crev \leq 0$ for all $\tldtau$ or equivalently, $\E[t~|~\tldt \geq \tldtau] - \frac{\crev}{\cP[\tldt \geq \tldtau]} \leq 0$ for all $\tldtau$ (including $\tldtau^*$ obtained in Case $1$). In this case, $\nrev^* = 0$. 
\end{enumerate}
So far, we have shown that $\nrev^* > 0$ if and only if $v^* = \frac{\sigma_t^2}{\sigma_{\tldt}}\cdot \frac{\phi\left(\tldtau^*/\sigma_{\tldt}\right)}{\Phi^c\left(\tldtau^*/\sigma_{\tldt}\right)} - \frac{\crev}{\Phi^c\left(\tldtau^*/\sigma_{\tldt}\right)} > 0$. We now need to show the following: 
\[
    v^* > 0 \iff \crev < \frac{1}{\sqrt{2\pi}}\cdot \frac{\sigma_t^2}{\sigma_{\tldt}}. 
\]
For the forward direction ($\implies$), we have: 
\begin{align*}
    v^* > 0 
    &\implies \frac{\sigma_t^2}{\sigma_{\tldt}} \cdot \phi(\tldtau^*/\sigma_{\tldt}) - \crev > 0 \quad \text{(since for finite $\tldtau^*$, $\frac{1}{\Phi^c(\tldtau^*/\sigma_{\tldt})} > 0$)} \\
    &\implies \crev < \frac{\sigma_t^2}{\sigma_{\tldt}}\cdot \phi(\tldtau^*/\sigma_{\tldt}) \\
    &\implies \crev < \frac{\sigma_t^2}{\sigma_{\tldt}}\cdot \phi(0) \quad \text{($\phi(x)$ is maximized at $x = 0$)} \\
    &\implies \crev < \frac{1}{\sqrt{2\pi}}\cdot \frac{\sigma_t^2}{\sigma_{\tldt}}. 
\end{align*}
For the other direction ($\impliedby$), we first consider the principal's net utility per applicant hired at threshold $\tldtau$, given by:
\[
    v(\tldtau) = \frac{1}{\Phi^c(\tldtau/\sigma_{\tldt})}\left[\frac{\sigma_t^2}{\sigma_{\tldt}}\cdot \phi(\tldtau/\sigma_{\tldt}) - \crev\right]. 
\]
Now,
\begin{align*}
    \crev < \frac{1}{\sqrt{2\pi}}\cdot \frac{\sigma_t^2}{\sigma_{\tldt}} &\implies \frac{\sigma_t^2}{\sigma_{\tldt}}\cdot \phi(0) - \crev > 0\\
    &\implies v(0) > 0 \\
    &\implies v^* > 0 \quad \text{(since $v^* = \max_{\tldtau}~v(\tldtau)$)}.
\end{align*}
Finally, using the simplified condition, we can characterize the optimal solution of the principal's optimization problem as follows: 
\begin{itemize}
    \item Solve for $\tldtau^* = arg\max_{\tldtau} \quad \frac{\sigma_t^2}{\sigma_{\tldt}}\cdot \frac{\phi\left(\tldtau/\sigma_{\tldt}\right)}{\Phi^c\left(\tldtau/\sigma_{\tldt}\right)} - \frac{\crev}{\Phi^c\left(\tldtau/\sigma_{\tldt}\right)}$. 
    \item If $\crev < \frac{1}{\sqrt{2\pi}}\cdot \frac{\sigma_t^2}{\sigma_{\tldt}}$, choose $\nrev^* = \frac{k}{\cP[\tldt \geq \tldtau^*]}$.
    \item Else, choose $\nrev^* = 0$. 
\end{itemize}

\subsubsection*{Part $3$: Showing that $\tldtau^*$ is unique and can be computed efficiently. }
In order to complete the last part of the proof, we will introduce the following claim: 
\begin{claim}\label{clm:tau_sol}
$\tldtau^*$, which maximizes the objective in ~(\ref{exp:opt_tau}), is the unique solution to the following non-linear equation:
\begin{align}\label{exp:nonlinear_eq}
    g\left(\frac{\tldtau}{\sigma_{\tldt}}\right) :=  \frac{\sigma_t^2}{\sigma_{\tldt}}\cdot \left[\phi\left( \frac{\tldtau}{\sigma_{\tldt}}\right) - \frac{\tldtau}{\sigma_{\tldt}} \cdot \Phi^c\left( \frac{\tldtau}{\sigma_{\tldt}} \right) \right] - \crev = 0.
\end{align}
In particular, $g(\cdot)$ is monotonically decreasing and $\tldtau^*$ is finite and can be obtained easily using binary search.  
\end{claim}
\begin{proof}
Recall that:
\[
    \tldtau^* = arg\max_{\tldtau} \quad \underbrace{\frac{\frac{\sigma_t^2}{\sigma_{\tldt}}\cdot \phi(\tldtau/\sigma_{\tldt})-\crev}{\Phi^c(\tldtau/\sigma_{\tldt})}}_{v(\tldtau)}, 
\]
Define $a = \frac{\sigma_t^2}{\sigma_{\tldt}}$ and $z = \frac{\tldtau}{\sigma_{\tldt}}$. Using this transformation, we have rewrite $v(\tldtau)$ as $w(z)$ where:
\[
     w(z) = \frac{a\cdot \phi(z)-\crev}{\Phi^c(z)},
\]
Now, if $z^* = arg\max_{z} ~w(z)$, it should be immediately clear that $\tldtau^* = \sigma_{\tldt}\cdot z^*$. Therefore, our goal now is to show that $z^*$ is the unique solution to $g(z) := a \cdot \phi(z) - a \cdot z \cdot \Phi^c(z) - \crev = 0$. Note that $w(z)$ is differentiable in $z$, therefore using the first order conditions, we conclude that any local extremum should satisfy: 
\begin{align*}
    w'(z) = 0 
    \implies &\frac{ a\cdot\Phi^c(z)\cdot \phi'(z) + \phi(z)\left(a \cdot \phi(z) - \crev \right) }{ \left( \Phi^c(z)\right)^2 } = 0 \\
    \implies & \frac{-a\cdot z \cdot \Phi^c(z)\cdot \phi(z) + \phi(z)\left(a \cdot \phi(z) - \crev \right) }{ \left( \Phi^c(z)\right)^2 } = 0  \quad \text{(since $\phi'(z) = -z\cdot \phi(z)$)} \\
    \implies & \frac{\phi(z)}{\left( \Phi^c(z)\right)^2 }\cdot \left[ -a\cdot z \cdot \Phi^c(z) + a\cdot \phi(z) - \crev \right] = 0 \\
    \implies & \frac{\phi(z)}{\left( \Phi^c(z)\right)^2 }\cdot g(z) = 0. 
\end{align*}
We claim that $\pm\infty$ cannot be the maximizers of $w(z)$. This follows from the fact that $\lim_{z \to +\infty}w(z) = -\infty$ and $\lim_{z \to -\infty}w(z) = -\crev$, but $u(0) > -\crev$. Therefore, the maximizer of $w(z)$ must be finite which immediately implies that the only possible maximizer of $w(z)$ must be a solution of $g(z) = 0$. \\

Now, $g(z)$ is continuous in $z$ and $\lim_{z \to\infty}g(z) > 0$ and $\lim_{z \to -\infty}g(z) < 0$. Therefore, by the intermediate value theorem, there must be a finite solution to $g(z) = 0$. Additionally, we have: 
\begin{align*}
    g'(z) &= a \cdot \phi'(z) - a \cdot \Phi^c(z) - a \cdot z \cdot (-\phi(z))\\
    &= - a \cdot z \cdot \phi(z) -  a \cdot \Phi^c(z) + a \cdot z \cdot \phi(z) \quad (\text{using $\phi'(z) = -z\cdot \phi(z)$}) \\
    &= -  a \cdot \Phi^c(z) < 0. 
\end{align*}
This means that $g(z)$ is monotonically decreasing in $z$ which implies the following:
\begin{itemize}
    \item $g(z) = 0$ has a unique solution $z^*$;
    \item $w'(z) > 0$ for $z < z^*$ and $w'(z) < 0$ for $z > z^*$, showing that $w(z)$ is concave and $z^*$ is the unique maximizer of $w(z)$; and 
    \item $z^*$ can be obtained using binary search. 
\end{itemize}
This concludes the proof of the claim. 
\end{proof}

\subsection{Equivalent Characterizations of Corollary~\ref{corr:viable}}\label{pf:corr_viable}
\begin{claim}\label{clm:tau_positive}
The following statements are equivalent:
\begin{enumerate}[label=\alph*)]
    \item $\tldtau^* > 0$,
    \item $\nrev^* > 0$,
    \item $\crev < \frac{1}{\sqrt{2\pi}}\cdot \frac{\sigma_t^2}{\sigma_{\tldt}}$.
\end{enumerate}
\end{claim}
\begin{proof}
In order to complete the proof, we will show that $c \implies b$, $b \implies a$ and $a \implies c$ in that order. 
\begin{enumerate}
\item The first part (showing $c \implies b$) follows directly from the statement of Theorem~\ref{thm:opt_n_tau}. 

\item For the second part ($b \implies a$), suppose that $\nrev^* > 0$. Therefore, it must be that $v^* = v(\tldtau^*)$ must be $ > 0$. This implies: 
\[
    \frac{1}{\Phi^c(\tldtau^*/\sigma_{\tldt})}\cdot \left[ \frac{\sigma_t^2}{\sigma_{\tldt}}\cdot \phi(\tldtau^*/\sigma_{\tldt}) - \crev \right] > 0, 
\]
which implies that $\frac{\sigma_t^2}{\sigma_{\tldt}}\cdot \phi(\tldtau^*/\sigma_{\tldt}) - \crev > 0$. But, we know that $g(\tldtau^*/\sigma_{\tldt}) = 0$ which means that: 
\[
        \frac{\sigma_t^2}{\sigma_{\tldt}}\cdot \phi(\tldtau^*/\sigma_{\tldt}) - \crev = \frac{\sigma_t^2}{\sigma_{\tldt}}\cdot \frac{\tldtau^*}{\sigma_{\tldt}}\cdot \Phi^c(\tldtau^*/\sigma_{\tldt}).
\]
Then, $\tldtau^* \cdot \Phi^c(\tldtau^*/\sigma_{\tldt}) > 0$ which clearly implies that $\tldtau^* > 0$. 

\item Finally, the first part of the proof ($a \implies c$), note that $\tldtau^* > 0$ implies that $g\left( \frac{\tldtau^*}{\sigma_{\tldt}} \right) < g(0) = \frac{1}{\sqrt{2\pi}}\cdot \frac{\sigma_t^2}{\sigma_{\tldt}} - \crev$ since $g(\cdot)$ is monotonically decreasing in its argument. But, $g\left( \frac{\tldtau^*}{\sigma_{\tldt}} \right) = 0$. Therefore, it must be that $\crev < \frac{1}{\sqrt{2\pi}}\cdot \frac{\sigma_t^2}{\sigma_{\tldt}}$. This concludes the proof.
\end{enumerate}
\end{proof}

\subsection{Proof of Lemma~\ref{lem:decent_util}}\label{pf:lem_decent_util}
The proof of the lemma follows directly from Claim~\ref{clm:tau_sol}. Recall that if $\tldtau^*$ is the optimal threshold for the principal, then it must satisfy $g(\tldtau^*/\sigma_{\tldt}) = 0$. This implies:
\begin{align*}
    &\frac{\sigma_t^2}{\sigma_{\tldt}} \cdot \phi\left(\frac{\tldtau^*}{\sigma_{\tldt}} \right) - \frac{\sigma_t^2}{\sigma_{\tldt}^2}\cdot \tldtau^* \cdot \Phi^c\left( \frac{\tldtau^*}{\sigma_{\tldt}}\right)  - \crev = 0\\
    \implies & \frac{\sigma_t^2}{\sigma_{\tldt}} \cdot \phi\left(\frac{\tldtau^*}{\sigma_{\tldt}} \right) - \crev = \frac{\sigma_t^2}{\sigma_{\tldt}^2}\cdot \tldtau^* \cdot \Phi^c\left( \frac{\tldtau^*}{\sigma_{\tldt}}\right) \\
    \implies &  \frac{\sigma_t^2}{\sigma_{\tldt}} \cdot \frac{ \phi\left( \tldtau^*/\sigma_{\tldt} \right)  }{ \Phi^c\left( \tldtau^*/\sigma_{\tldt} \right) } - \frac{\crev}{\Phi^c\left( \tldtau^*/\sigma_{\tldt} \right) } = \frac{\sigma_t^2}{\sigma_{\tldt}^2}\cdot \tldtau^*.
\end{align*}
The LHS represents the principal's utility per hired applicant at threshold $\tldtau^*$. This concludes the proof of the lemma.

\subsection{Monotonicity Results}\label{app:monotonicity}
\begin{claim}\label{clm:tau_n_dep_c}
If $c_{rev,1} > c_{rev,2}$, then $\tldtau_1^* < \tldtau_2^*$ and $n_{rev,1}^* \leq n_{rev,2}^*$.
\end{claim}
\begin{proof}
The order of optimal thresholds follows directly from noting that the function $g\left(\frac{\tldtau}{\sigma_{\tldt}}\right)$ is monotonically decreasing in $\tldtau$. Now, there are $3$ cases:
\begin{enumerate}
    \item $0 < \tldtau_1^* < \tldtau_2^* \implies n_{rev,1}^* < n_{rev,2}^*$ since $\nrev^* = \frac{k-\nold}{\cP[\tldt \geq \tldtau^*]}$.
    \item $\tldtau_1^* \leq 0 < \tldtau_2^* \implies 0 = n_{rev,1}^* < n_{rev,2}^*$ (by Claim~\ref{clm:tau_positive}).
    \item $\tldtau_1^* < \tldtau_2^* \leq 0 \implies n_{rev,1}^* = n_{rev,2}^* = 0$.
\end{enumerate} 
This concludes the proof.
\end{proof}

\begin{claim}\label{clm:tau_dep_sigmat}
Suppose, $\crev$ is small enough that the optimal threshold is positive. Then, a higher value of $\sigma_t$ leads to a higher optimal threshold $\tldtau^*$, i.e., $\sigma_{t,1} > \sigma_{t,2} \implies \tldtau_1^* > \tldtau_2^*$.    
\end{claim}
\begin{proof}
Recall that $\tldtau^*$ is obtained as the unique solution to $g\left( \frac{\tldtau^*}{\sigma_{\tldt}}\right) = 0$, which implies that: 
\[
     \phi\left(\frac{\tldtau^*}{\sigma_{\tldt}} \right) - \frac{\tldtau^*}{\sigma_{\tldt}}\cdot \Phi^c \left(\frac{\tldtau^*}{\sigma_{\tldt}} \right) = \frac{\sigma_{\tldt}}{\sigma_t^2}\cdot \crev.
\]
As $\sigma_t$ increases, the ratio $\frac{\sigma_{\tldt}}{\sigma_t^2} = \frac{\sqrt{\sigma_t^2 + \alpha^2 \sigma_{ef}^2 + (1-\alpha)^2\sigma_{es}^2}}{\sigma_t^2}$ decreases and therefore, the RHS of the above equation decreases. Since we know that the LHS is monotonically decreasing in the ratio $\frac{\tldtau^*}{\sigma_{\tldt}}$, for the equality to still hold, the ratio should increase. Therefore, when $\sigma_{t,1} > \sigma_{t,2}$, we have: 
\[
     \frac{\tldtau_1^*}{\sigma_{\tldt,1}} > \frac{\tldtau_2^*}{\sigma_{\tldt,2}} \implies \frac{\tldtau_1^*}{\tldtau_2^*} > \frac{\sigma_{\tldt,1}}{\sigma_{\tldt,2}} > 1 \implies \tldtau_1^* > \tldtau_2^*.  
\]
\end{proof}

\subsection{Proof of Lemma~\ref{lem:decent_adv_better_util}}\label{pf:lem_decent_adv_better_util}
We will complete the proof separately for the two different models of disadvantage for group $B$.

\paragraph{Negatively biased signal mean.} We consider that Group $B$ is disadvantaged in the sense that the mean of the distribution of signal $\tilde s$ for group $B$ is negatively biased by amount $\beta$ with respect to Group $A$. The bias in $\tilde s$ translates into a bias $\beta'$ in the mean of the perceived quality ($\tilde t$) distribution for Group $B$ with $\beta' = (1-\alpha)\beta$, i.e., $\tilde t$ for Group $A$ follows $\mathcal{N}(0, \sigma_{\tldt}^2)$, while for group $B$ applicants, it follows $\mathcal{N}(-\beta', \sigma_{\tldt}^2)$. We first derive the revised expression for the net expected utility for the principal for hiring an applicant from group $B$ above some threshold $\tldtau$. We have: 
\begin{align*}
    v_B(\tldtau) &= \mathbb{E}\left[t~|~\tldt \geq \tldtau \right] - \frac{\crev}{\mathbb{P}\left[\tldt \geq \tldtau \right]} \\
    &= \frac{\sigma_t^2}{\sigma_{\tldt}^2}\cdot \mathbb{E}\left[\tldt + \beta'~|~ \tldt \geq \tldtau \right] - \frac{\crev}{\mathbb{P}\left[\tldt \geq \tldtau \right]} \\
    &= \frac{\sigma_t^2}{\sigma_{\tldt}^2}\cdot \beta' + \frac{\sigma_t^2}{\sigma_{\tldt}^2}\cdot \left(-\beta' + \sigma_{\tldt}\cdot H\left(\frac{\tldtau + \beta'}{\sigma_{\tldt}} \right) \right) - \frac{\crev}{ \Phi^c\left(\frac{\tldtau+\beta'}{\sigma_{\tldt}} \right) } \\
    &= \left[  \frac{\sigma_t^2}{\sigma_{\tldt}} \cdot  H\left(\frac{\tldtau + \beta'}{\sigma_{\tldt}} \right) - \frac{\crev}{ \Phi^c\left(\frac{\tldtau+\beta'}{\sigma_{\tldt}} \right) }   \right]
\end{align*}
Now, suppose, $\tldtau_A^*$ maximizes $v_A(\tldtau) = \frac{\sigma_t^2}{\sigma_{\tldt}} \cdot  H\left(\frac{\tldtau }{\sigma_{\tldt}} \right) - \frac{\crev}{ \Phi^c\left(\frac{\tldtau}{\sigma_{\tldt}} \right) }$. This implies that $\tldtau_A^* - \beta'$ must maximize $v_B(\tldtau)$, i.e., $\tldtau_B^* = \tldtau_A^* - \beta'$. \\
\noindent 
This immediately implies that:
\[
    v_B(\tldtau_B^*) = v_B(\tldtau_A^*-\beta') = v_A(\tldtau_A^*). 
\]
This concludes the proof for the first part. 

\paragraph{Noisier signal distribution, i.e., $\sigma_{e, A} < \sigma_{e, B}$.} The idea for this part of the proof is to use the simplified expression for the principal's optimal utility at equilibrium using Lemma~\ref{lem:decent_util}. We have: 
\[
   v^* = v(\tldtau^*) = \E[t~|~\tldt \geq \tldtau^*] - \frac{\crev}{\cP[\tldt \geq \tldtau^*]} = \frac{\sigma_t^2}{\sigma_{\tldt}}\cdot \frac{\phi(\tldtau^*/\sigma_{\tldt})}{\Phi^c(\tldtau^*/\sigma_{\tldt})} - \frac{\crev}{\Phi^c(\tldtau^*/\sigma_{\tldt})} = (\sigma_t^2)\cdot \frac{1}{\sigma_{\tldt}}\cdot \left(\frac{\tldtau^*}{\sigma_{\tldt}} \right).
\]

Now, recall that $\sigma_{t_A} = \sigma_{t_B}$. But, $\sigma_{e,A} < \sigma_{e,B} \implies \sigma_{\tldt_A} < \sigma_{\tldt_B} \implies \frac{1}{\sigma_{\tldt_A}} > \frac{1}{\sigma_{\tldt_B}}$. Now we will argue that $\frac{\tldtau_A^*}{\sigma_{\tldt_A}} > \frac{\tldtau_B^*}{\sigma_{\tldt_B}}$. From Claim~\ref{clm:tau_sol}, we know that $\frac{\tldtau^*}{\sigma_{\tldt}}$ must satisfy: 
\[
           \phi\left( \frac{\tldtau^*}{\sigma_{\tldt}} \right) - \frac{\tldtau^*}{\sigma_{\tldt}}\cdot \Phi^c\left( \frac{\tldtau^*}{\sigma_{\tldt}}\right) = \frac{\sigma_{\tldt}}{\sigma_t^2}\cdot \crev.
\]
For Group $B$, the RHS of the above equation is larger (due to larger $\sigma_{\tldt}$). Since the LHS is monotonically decreasing in $\frac{\tldtau^*}{\sigma_{\tldt}}$, it must be that $\frac{\tldtau_A^*}{\sigma_{\tldt_A}} > \frac{\tldtau_B^*}{\sigma_{\tldt_B}}$.
Putting all parts together, we conclude that $v_A^* > v_B^*$.

\subsection{Proof of Theorem~\ref{thm:no_hire_B}}\label{app:no_hire_B}

Note that the above optimization problem can be decoupled as follows: Pick $r_A$ such that $0 \leq r_A \leq k$ and then pick $0 \leq r_B \leq k - r_A$. Then, we can decompose the constraint as follows: 
\[
   \nrev(A)\cdot \cP[\tldt_A \geq \tldtau_A] \leq r_A; \quad \nrev(B)\cdot \cP[\tldt_B \geq \tldtau_B] \leq r_B, 
\]
which leads to the following independent optimization problems for groups $A$ and $B$. 
\begin{align*}
      \max_{\nrev(A), \tldtau_A} \quad &\nrev(A)\cdot \cP[\tldt_A \geq \tldtau_A] \cdot \E[t_A~|~\tldt_A \geq \tldtau_A] - \crev \nrev(A) \quad s.t. \\
      & \nrev(A) \cdot \cP[\tldt_A \geq \tldtau_A] \leq r_A, \nrev(A) \geq 0. 
\end{align*}
\begin{align*}
      \max_{\nrev(B), \tldtau_B} \quad &\nrev(B)\cdot \cP[\tldt_B \geq \tldtau_B] \cdot \E[t_B~|~\tldt_B \geq \tldtau_B] - \crev \nrev(B) \quad s.t. \\
      & \nrev(B) \cdot \cP[\tldt_B \geq \tldtau_B] \leq r_B, \nrev(B) \geq 0. 
\end{align*}
We know how to solve either of these problems in isolation. The assumption on $\crev$ guarantees that both problems have positive objective values at optimality. Let the optimal thresholds be given independently as $\tldtau_A^*$ and $\tldtau_B^*$. In that case, $\left( \frac{r_A}{\cP[\tldt_A \geq \tldtau_A^*]}, \frac{r_B}{\cP[\tldt_B \geq \tldtau_B^*]}, \tldtau_A^*, \tldtau_B^* \right)$ is a feasible solution to Problem~\ref{opt:joint}. Using an identical argument as in the proof of Theorem~\ref{thm:opt_n_tau}, we can show that at optimality, $r_A + r_B = k$. Thus, the problem reduces to choosing $r_A$ so that it maximizes the objective in~\ref{opt:joint}. Therefore, we have: 
\begin{align*}
   \max_{r_A \leq k} \quad r_A \cdot \E[t_A~|~\tldt_A \geq \tldtau_A^*] + (k- r_A) \cdot \E[t_B~|~\tldt_B \geq \tldtau_B^*] - \crev \cdot \left( \frac{r_A}{\cP[\tldt_A \geq \tldtau_A^*]} + \frac{k-r_A}{\cP[\tldt_B \geq \tldtau_B^*]} \right),
\end{align*}
which by re-arranging terms, we can rewrite as follows:
\begin{align*}
   \max_{r_A \leq k} \quad r_A \cdot \underbrace{\left[ \left(\E[t_A~|~\tldt_A \geq \tldtau_A^*] - \frac{\crev}{\Phi^c(\tldtau_A^*/\sigma_{\tldt_A})}\right) - \left(\E[t_B~|~\tldt_B \geq \tldtau_B^*] - \frac{\crev}{\Phi^c(\tldtau_B^*/\sigma_{\tldt_B})}\right)  \right]}_{(v_A^* - v_B^*)} \\
   + \underbrace{k \cdot \E[t_B~|~\tldt_B \geq \tldtau_B^*] - \crev \cdot \frac{k}{\cP[\tldt_b \geq \tldtau_B^*]}}_{\text{terms independent of $r_A$}}
\end{align*}
From Lemma~\ref{lem:decent_adv_better_util}, we know that $v_A^* > v_B^*$ which immediately shows that co-efficient of $r_A$ in the expression above is positive. Therefore, $r_A^* = k$ which implies that $r_B^* = 0$ or equivalently, $\nrev(B)^* = 0$. This concludes the proof.


\section{Proofs of Supplementary Results}

\subsection{Proof of Claim~\ref{tech_clm:hazard}}\label{app:hazard}
We know that $X \sim \mathcal{N}(\mu_X, \sigma_x^2)$. Now, 
\begin{align*}
    \E[X~|~X \geq a] &= \int_{a}^{\infty} u \cdot \cP[X = u~|~X \geq a]du\\
    &= \int_{a}^{\infty} u \cdot \frac{\cP[X = u, X \geq a]}{\cP[X \geq a]}du\\
    &= \frac{1}{\cP[X \geq a]} \int_{a}^{\infty} u \cdot \cP[X = u] du \\
    &= \frac{1}{\Phi^c\left(\frac{a-\mu_X}{\sigma_X} \right)} \int_{a}^{\infty} u \cdot \frac{1}{\sigma_X\sqrt{2\pi}}\cdot \exp\left(- \frac{(u-\mu_X)^2}{2\sigma_X^2}\right) du. 
\end{align*}
Now we do a variable substitution. Define $w = \frac{u-\mu_X}{\sigma_X}$. Then, we can rewrite as follows: 
\begin{align*}
    \E[X~|~X \geq a] &= \frac{1}{\Phi^c\left(\frac{a-\mu_X}{\sigma_X} \right)} \int_{(a-\mu_X)/\sigma_X}^{\infty} (\mu_X + \sigma_X\cdot w)\frac{1}{\sqrt{2\pi}}\exp\left(-\frac{w^2}{2}\right)dw \\
    &= \frac{1}{\Phi^c\left(\frac{a-\mu_X}{\sigma_X} \right)} \left[ \mu_X  \int_{(a-\mu_X)/\sigma_X}^{\infty}\frac{1}{\sqrt{2\pi}}\exp\left(-\frac{w^2}{2}\right)dw + \sigma_X \int_{(a-\mu_X)/\sigma_X}^{\infty}\frac{w}{\sqrt{2\pi}}\exp\left(-\frac{w^2}{2}\right)dw  \right]\\
    &= \frac{1}{\Phi^c\left(\frac{a-\mu_X}{\sigma_X} \right)} \left[ \mu_X \cdot \Phi^c\left(\frac{a-\mu_X}{\sigma_X} \right) + \sigma_X\int_{(a-\mu_X)/\sigma_X}^{\infty}w \phi(w)dw  \right] \\
    &= \mu_X + \frac{\sigma_X}{\Phi^c\left(\frac{a-\mu_X}{\sigma_X} \right)} \int_{(a-\mu_X)/\sigma_X}^{\infty}w \phi(w)dw.
\end{align*}
Note that the proof is completed if we can just show that the integrand above equals $\phi\left(\frac{a-\mu_X}{\sigma_X} \right)$.
\begin{align*}
    \int_{(a-\mu_X)/\sigma_X}^{\infty}w \phi(w)dw &= \int_{(a-\mu_X)/\sigma_X}^{\infty}-\phi'(w)dw \quad \text{(since $\phi'(w) = -w\phi(w)$)}\\
    &= \int_{\infty}^{(a-\mu_X)/\sigma_X} \phi'(w) dw \\
    &= \phi\left(\frac{a-\mu_X}{\sigma_X}\right) - \phi(0) \\
    &= \phi\left(\frac{a-\mu_X}{\sigma_X}\right).
\end{align*}

\subsection{Monotonicity of the standard normal hazard rate function $H(x)$}\label{app:monotone_hazard}
Here, we show that hazard rate function of the standard normal random variable, given by $H(x) = \frac{\phi(x)}{\Phi^c(x)}$ is monotonically increasing in $x$. Observe that $H(x) \geq 0$ trivially for all $x$. First, we will show that $H(x) \geq x$ for all $x$. It suffices to show for $x > 0$ (since it holds trivially for $x \leq 0$). 
\begin{align*}
    1 - \Phi(x) &= \int_{u = x}^{\infty} \frac{1}{\sqrt{2\pi}}\exp\left(-\frac{u^2}{2} \right)du \\
    &\leq \int_{u = x}^{\infty} \frac{1}{\sqrt{2\pi}}\frac{u}{x}\exp\left(-\frac{u^2}{2} \right)du \quad \text{(since $\frac{u}{x} \geq 1$)} \\
    &= \frac{1}{x\sqrt{2\pi}} \exp\left(-\frac{x^2}{2}\right) \\
    &= \frac{\phi(x)}{x}.
\end{align*}
Therefore, $H(x) = \frac{\phi(x)}{1-\Phi(x)} \geq x$. Now, 
\begin{align*}
    H'(x) &= \frac{ \Phi^c(x) \phi'(x) + (\phi(x))^2 }{\left( \Phi^c(x) \right)^2 } \\
    &= \frac{-x \phi(x) \Phi^c(x) + (\phi(x))^2}{(\Phi^c(x))^2} \quad \text{(substituting $\phi'(x) = -x\phi(x)$)} \\
    &= -x H(x) + (H(x))^2 \\
    &= H(x)(H(x) - x) \geq 0.
\end{align*}
The last equality follows from the fact that $H(x)\geq 0$ and $H(x) \geq x$. This concludes the proof that $H(x)$ is increasing. 


\section{Additional Figures}\label{app:extra_fig}

\begin{figure}[!ht]
\centering 
\includegraphics[width=0.5\textwidth]{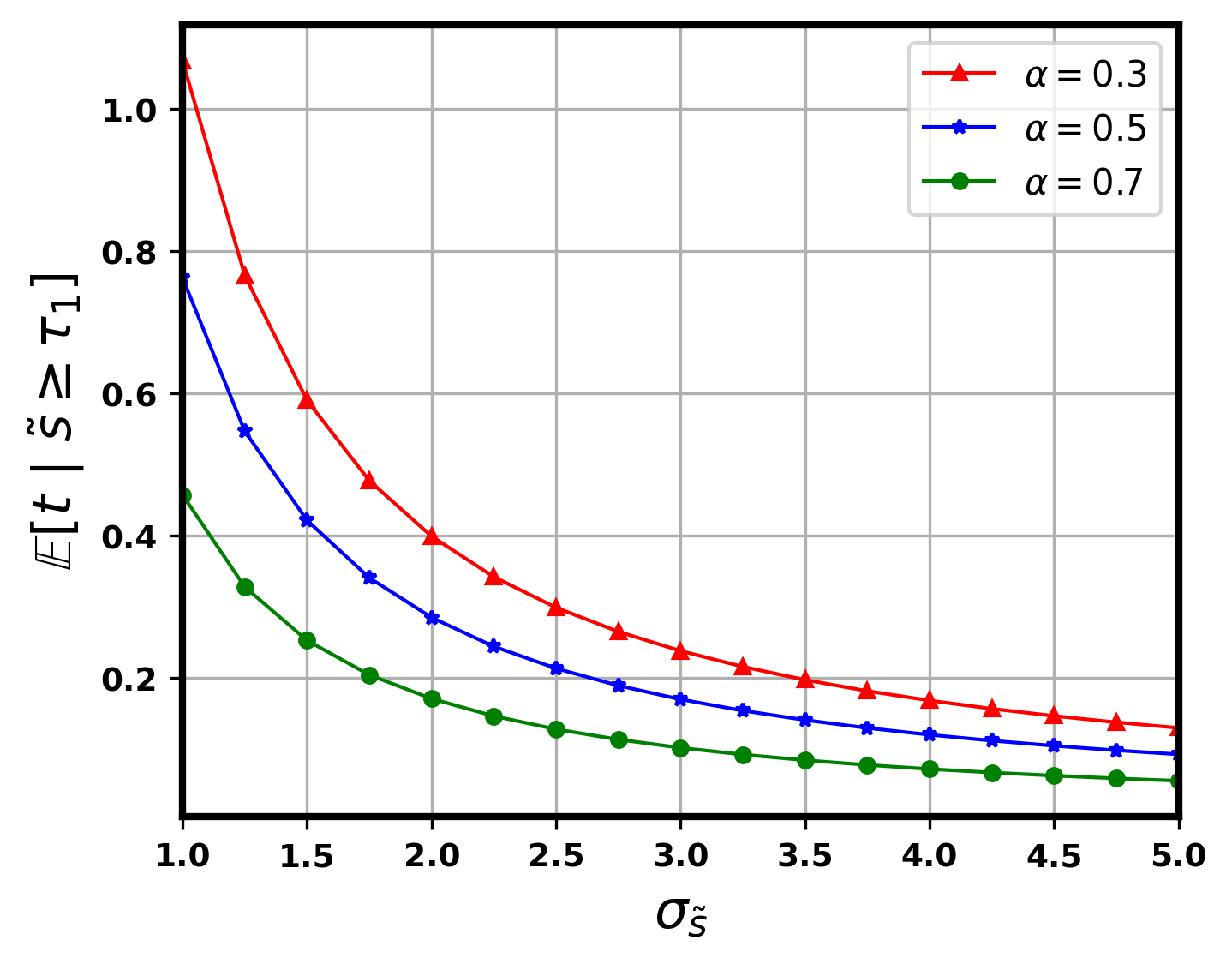}
\caption{We plot the expected quality of a selected applicant (when the agent uses a selection threshold of $\tau_1$) as a function of the variance $\sigma_{\tilde s}$ of the noisy signal $\tilde s$. As $\sigma_{\tilde s}$ increases, the expected quality decays monotonically for all $\alpha$.}
\label{fig:quality_tlds}
\end{figure}

\begin{figure}[!h]
    \centering
    \subfloat[]{\includegraphics[width=0.44\textwidth]{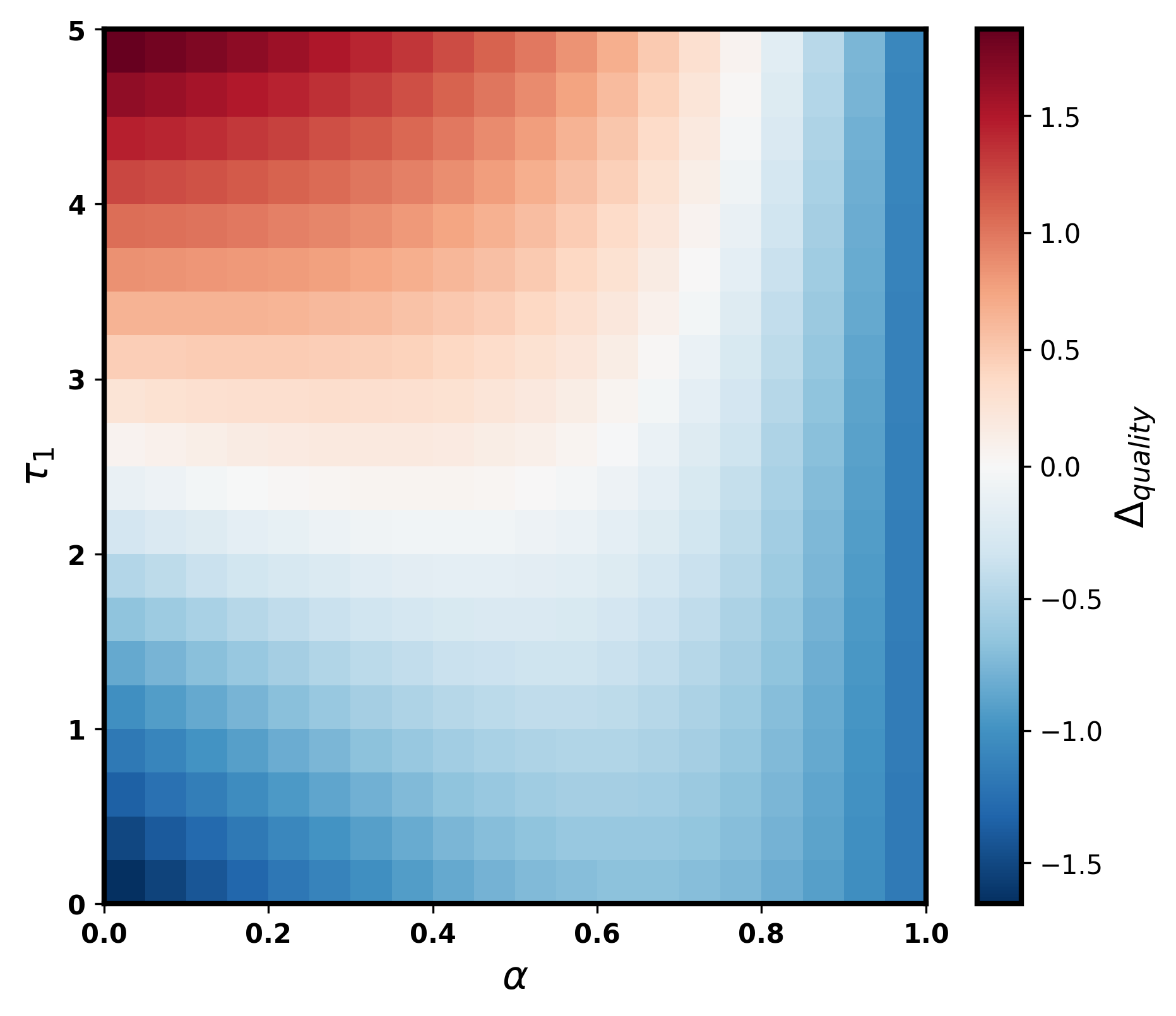}}
    \subfloat[]{\includegraphics[width=0.44\textwidth]{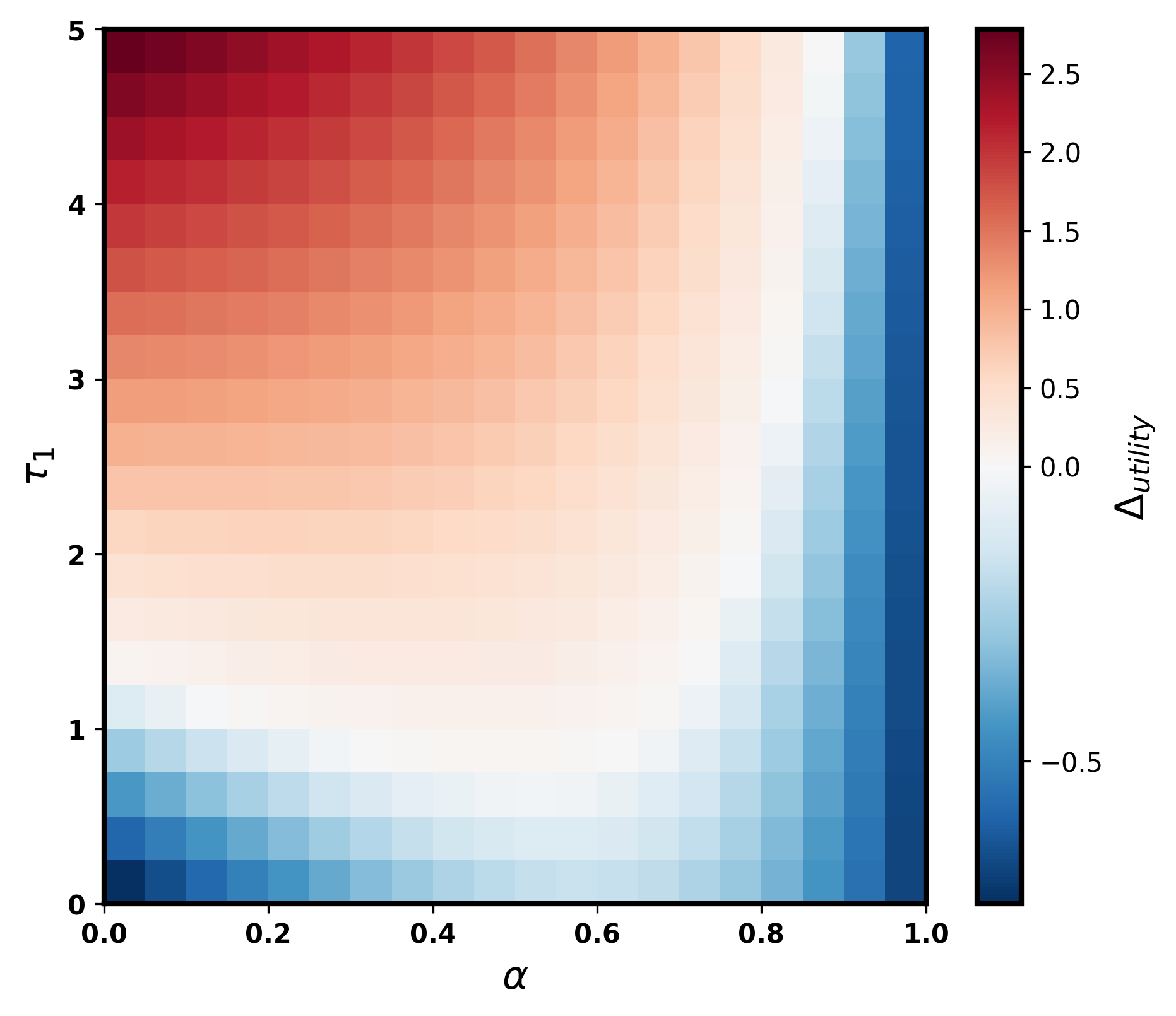}}  
    \caption{We plot a heatmap of values of $\Delta_{quality}$ (left) and $\Delta_{utility}$ (right) over different combinations of the principal's weight $\alpha$ over the interval $[0,1]$ and the agent's threshold $\tau_1$ over the interval $[0, 5]$. Each heatmap consists of $400$ grid points with each grid of the size $0.05 \times 0.25$. Values are evaluated at the centre of each grid. The red regions indicate locations where $\Delta > 0$ and blue regions indicate locations where $\Delta < 0$. The white regions indicate the transition boundary where $\Delta \approx 0$. The heatmaps validate our discussion about where delegation is beneficial for the principal (Section~\ref{sec:comparison}). }
    \label{fig:heatmaps}
\end{figure}

\end{document}